\documentclass[journal]{IEEEtran}

\usepackage{cite}
\usepackage{amsmath, amssymb,amsfonts}
\usepackage{algorithmic}
\usepackage{graphicx}
\usepackage{textcomp}
\usepackage{xcolor}
\usepackage{dsfont} 
\usepackage{mathrsfs}
\usepackage{braket}
\usepackage{xfrac}
\usepackage{booktabs}
\usepackage{multirow}
\usepackage[scr=boondoxo,scrscaled=1.05]{mathalfa}
\def\BibTeX{{\rm B\kern-.05em{\sc i\kern-.025em b}\kern-.08em
T\kern-.1667em\lower.7ex\hbox{E}\kern-.125emX}}

\newcommand{\overbar}[1]{\mkern 1.5mu\overline{\mkern-1mu#1\mkern-1mu}\mkern 1.5mu}
\newcommand{\underbared}[1]{\mkern 1.5mu\underline{\mkern-1.5mu#1\mkern-1.5mu}\mkern 1.5mu}

\DeclareMathAlphabet\mathbfcal{OMS}{cmsy}{b}{n}
\DeclareMathAlphabet{\mathbfscr}{OMS}{mdugm}{b}{n}
\newcommand{\part}{\overbar{\mathbf{X}}}
\newcommand{\partcover}[1]{\mathbf{C} \left( \overbar{\mathbf{X}}^{#1} \right)}
\newcommand{\hashfamily}{\underbared{\mathbf{B}}}
\newcommand{\randhashfamily}{\mathbfcal{B}}
\newcommand{\hashfamilyset}{\mathbfcal{B}}

\let\originalleft\left
\let\originalright\right
\renewcommand{\left}{\mathopen{}\mathclose\bgroup\originalleft}
\renewcommand{\right}{\aftergroup\egroup\originalright}

\newtheorem{thm}{Theorem}
\newtheorem{lem}{Lemma}

\newtheorem{prop}{Proposition}
\newtheorem{defn}{Definition}

\title{Minimum Feedback for Collision-Free Scheduling in Massive Random Access}

\begin{document}

\author{Justin Kang,~\IEEEmembership{Student Member, IEEE}, 
Wei Yu,~\IEEEmembership{Fellow, IEEE}\thanks{Manuscript submitted to IEEE Transactions on Information Theory on July 29, 2020, revised on
May 11, 2021.
This work is supported by Natural Science and Engineering Research Council (NSERC).
The materials in this paper have been presented in part at the {\it IEEE
International Symposium on Information Theory (ISIT)}, June 2020.
The authors are with The Edward S. Rogers Sr. Department of Electrical 
and Computer Engineering, University of Toronto, Toronto, ON M5S 3G4, Canada.
(e-mails: jkang@ece.utoronto.ca, weiyu@ece.utoronto.ca.)}}
\maketitle

\begin{abstract}
Consider a massive random access scenario in which a small set of $k$
active users out of a large number of $n$ potential users need to be
scheduled in $b\ge k$ slots.  What is the minimum common feedback to the
users needed to ensure that scheduling is collision-free? 
Instead of a naive scheme of listing the indices of the $k$ active
users in the order in which they should transmit, at a cost of
$k\log(n)$ bits, this paper shows that for the case of $b=k$, the 
rate of the minimum fixed-length common feedback code scales only 
as $k \log(e)$ bits, plus an additive term that scales in $n$ as 
$\Theta \left(\log \log(n) \right)$ for fixed $k$. 
If a variable-length code can be used, assuming uniform activity among
the users, the minimum average common feedback rate still requires 
$k \log(e)$ bits, but the dependence on $n$ can be reduced to
$O(1)$.  When $b>k$, the number of feedback bits needed for
collision-free scheduling can be significantly further reduced.
Moreover, a similar scaling on the minimum feedback rate is 
derived for the case of scheduling $m$ users per slot, when $k \le mb$.
The problem of constructing a minimum collision-free feedback
scheduling code is connected to that of constructing a perfect hashing
family, which allows practical feedback scheduling codes to be constructed from
perfect hashing algorithms.

 \end{abstract}
\IEEEpeerreviewmaketitle


\section{Introduction} \label{sec:introduction}

\IEEEPARstart{T}{his} paper is motivated by the massive connectivity
scenario for machine-type wireless communications, in which a
base-station (BS) needs to provide connectivity to a massive number of $n$
devices (e.g., in the order of $10^5\sim 10^6$), but their traffic is
sporadic so that at any given time, only a random subset of $k \ll n$ 
users are active \cite{XChenCapacityManyAccess, ChenMassiveAccess}. 
We assume the following random access scheme involving three phases with
rate-limited feedback. In the first phase, $k$ active users transmit
pre-assigned uniquely identifying pilot sequences over the uplink
multiple access channel to indicate their activities. The BS uses a
multiuser detection algorithm, typically involving compressed sensing
\cite{Wunder2015SparseSignal, Chen2018SparseConnectivity,
Fengler2019NonBaysian, amalladinne2018coded}, to determine the active
set of users. In the second phase, the BS transmits a common feedback
message to all the active users over a noiseless downlink broadcast
channel; this feedback message specifies a schedule for the subsequent
data transmissions of the $k$ active users over $b$ orthogonal slots.  
In the third phase, the users transmit their payload data over the
scheduled slots.

Assuming that in the first phase the user activities
are detected perfectly, this paper focuses on the second scheduling phase
and asks the question: What is the minimum rate of the common
feedback message from the BS to the $k$ users so that the scheduling over
the $b$ slots is collision free, assuming $b \ge k$? 
An extension of the above question is the following: If $(m-1)b < k \le mb$ for
some positive integer $m$, what is the minimum feedback needed to
ensure that at most $m$ users are scheduled in each slot? 

A naive feedback scheme for collision-free scheduling is to index 
each of the $n$
users, then let the BS list off all the active user indices in the
order that they should transmit.  This requires a feedback message of
$k\log\left(n \right)$ bits. When the number of potential users is
large, the $\log(n)$ factor can be significant (e.g., 20 bits for
$n=10^6$), especially if the subsequent payload size is small, as
typically is the case for machine-type communications. 

This naive scheme is far from optimal. The results of this paper show that in the case where the feedback codewords have a fixed length regardless of
the user activity pattern, the minimum feedback rate for scheduling $k$ 
active users in $k$ slots with no collision can be reduced to approximately $k \log(e)$ 
bits plus a term that scales in $n$ as $\Theta(\log\log(n))$ bits for fixed $k$.
The above result can actually already be inferred from classic combinatorics
literature, because the problem of minimum feedback for collision-free
scheduling can be shown to be intimately connected to the hypergraph covering
and perfect hashing problems. 
This paper further shows that if the feedback
codeword length can be variable, then the scheduling overhead can be further reduced
to $k \log (e) + O(1)$ bits in expectation regardless of the probability distribution of the user activities. 
This is surprising as it implies that $k$ active users out of
$n$ potential users can be scheduled with a feedback rate of essentially 
$\log(e)$ bits (or 1 nat) per active user for arbitrarily large $n$, in
contrast to the optimal fixed-length code (or the naive scheme), which 
has a unbounded rate as $n$ tends to infinity. 
The above results can be extended to the $b>k$ case where the feedback
overhead can be significantly further reduced. 
Moreover, this paper also investigates 
the case where $b<k$ and multiple users need to be scheduled 
in each slot, which also requires less feedback as compared to the $b=k$ case. 

The \emph{scheduled} approach to random access considered in this paper can
be compared with \emph{contention} based schemes such as slotted ALOHA
\cite{Roberts1975ALOHACapture}, which, due to collision and
retransmission, has an overhead of roughly $Bk\left(\frac{1}{\eta} -
1\right)$ bits, where $B$ is the payload size and $\eta$ is the
efficiency of the chosen ALOHA variant. This efficiency varies from $\eta =
\frac{1}{e} \approx 0.37$ for classic slotted ALOHA, to $\eta \approx
0.8$ in irregular repetition slotted ALOHA
\cite{Liva2011Graph-BasedALOHA}, but can theoretically be taken arbitrarily
 close to 1 with the optimal coding scheme described in \cite{Narayanan2012}. 
The scheduled approach to random access can also be compared to \emph{unsourced}
multiple access in which the active users are detected only up to their permutations
\cite{Polyanskiy2017ARandom-access}. 
The unsourced approach can be useful for situations when users do not need to
identify themselves, for example, for fire detection.  Additionally, it also
serves as a way to abstract user identification to a higher level of the
communication protocol. For both unsourced multiple access as well as ALOHA, if
the user identifications are needed, then embedding such identification information
in the payload would cost $\log(n)$ bits. 

For the scheduled approach to massive random access, while the uplink activity
detection phase (e.g., using compressed sensing) would also generally have an
$O(\log (n))$ overhead, this paper points out that the scheduling cost in the 
downlink feedback phase can be shown to require approximately only $k \log(e)$ bits 
plus an $O(\log\log(n))$
overhead in the fixed-length case and $O(1)$ overhead in the variable-length case. 
This $\log\log(n)$ factor is reminiscent of the \emph{identification
capacity} \cite{Ahlswede1992IdentificationChannels}, but the use of 
a feedback scheduling code based on identification codes would have required 
$k \cdot \Theta(\log\log(n))$ bits, which is more costly than the feedback scheme
presented in this paper. Furthermore, the variable-length scheme proposed in
this paper has a feedback rate with
no dependence on $n$, which is a clear advantage in the regime where $n$ is large. 

The optimal feedback scheduling codes presented in this paper involve finding a
minimal family of partitions over $\left\{1, \dotsc ,n\right\}$ such that
no matter which user activity pattern occurs, there is always one
partition for which each subset of the partition contains exactly one
active user. For the case of the fixed-length code, this problem is
equivalent to the hypergraph covering problem
\cite{Snir1979TheHypergraphs} and the perfect hashing family problem
\cite{Fredman1984OnFunctions}. This allows the leveraging of classic
results in combinatorics to obtain upper and lower bounds on the
minimum feedback rate. Moreover, practical perfect hashing codes have
also been studied in the literature (e.g., \cite{Atici2000AFunction}). 
This allows us to use the code constructions in perfect hashing theory to
design practical feedback codes for optimal collision-free scheduling.

The rest of the paper is organized as follows. Section \ref{sec:setup} defines
the minimal feedback problem for collision-free scheduling, and the connection to
perfect hashing and hypergraph covering. We first focus on
the $b=k$ case and derive the achievable rates and lower bounds on the minimum
feedback in Sections \ref{sec:upper} and \ref{sec:lower}, respectively. These
results are extended to the case of $b>k$ in Section \ref{sec:tradeoff} and to
the case of $b<k$ in Section \ref{sec:multiuser}. Finally, 
construction of practical codes is discussed in Section
\ref{sec:practical}. 

The notations used in this paper are as follows. We use $\left[n \right]$ to
denote $\left\{ 1, 2, \dotsc, n\right\}$ and $\binom{\left[n \right]}{k}$ to
denote the set of all $k$-element subsets of $\left[n \right]$. All other sets
are typeset in upper case boldface, e.g. $\mathbf{X}$, while tuples of sets are typeset in upper case boldface with over-bar, e.g., $\part$,  and families of tuples are typeset with under-bar, e.g., $\hashfamily$. In addition, sets of families are typeset in boldface calligraphic letters, e.g., $\hashfamilyset$. We use $\log(\cdot)$ to denote logarithm in
base 2, and $\ln(\cdot)$ for natural logarithm in base $e$.
We use
$\mathbb{E}[\cdot]$ to denote the expected value of a random variable and $H(\cdot)$ to denote the entropy of a discrete random variable, or of a probability distribution. We use the shorthand $a^{\underline{b}} = a(a-1) \cdots (a-b+1)$. Furthermore, we use $X \sim \mathcal{U} \left( \mathbf{S}\right)$ denote a random variable which is uniformly distributed over the set $\mathbf{S}$.
 \section{Feedback for Collision-Free Scheduling} \label{sec:setup}
\subsection{Problem Formulation}

Assuming successful detection of $k$ active users among $n$ potential users at
the BS, the problem of designing a feedback code for scheduling the $k$ users
over $b$ slots can be cast as constructing an encoding function at
the BS that maps all possible occurrence of $k$-tuples out of $n$ users to a feedback message in an index set $[T]$:
\begin{equation}\label{eq:encoder}
f: \binom{\left[ n \right]}{k} \to [T]
\end{equation}
and designing a decoding function $g_i$ for each user $i$, 
which specifies each user's scheduled slot, i.e., 
\begin{equation} \label{eq:decoder}
g_i: [T] \to  \left[ b \right], 
\quad i \in \left[ n \right]. \end{equation}

For the case of $b \ge k$, we require the subsequent transmissions by the $k$
active users over the $b$ orthogonal slots to take place in a collision-free 
manner. Specifically, define an \emph{activity pattern} to be an element
$\mathbf{A} \in \binom{\left[ n \right]}{ k}$, which is a set of indices of 
$k$ active users. A feedback scheme is collision-free if
\begin{equation} \label{eq:seperation_1}
\forall \mathbf{A} \in \binom{\left[ n \right]}{ k}, \;\; \forall i \neq j \in
\mathbf{A}, \;\; 
g_i \left( f \left( \mathbf{A}\right)\right) \neq
g_j \left( f \left( \mathbf{A}\right) \right).
\end{equation}
In other words, a collision occurs whenever the decoding functions of two active
users within the same activity pattern produce the same output.

The above question can be extended to the case of $b < k$. 
Suppose $(m-1)b < k \le mb$ for some positive integer $m$. 
We define a similar condition requiring that at most $m$ users are scheduled in 
each of the $b$ orthogonal slots. Mathematically, this means that 
$\forall \mathbf{A} \in \binom{\left[ n \right]}{ k}$, we must have
\begin{equation} \label{eq:seperation_3}
\Big| \left\{ i \in \mathbf{A}, \;\; \text{s.t.}\;\; 
g_i \left( f \left( \mathbf{A}\right)\right) = d
\right\} \Big| \le m, \quad
 \forall d \in [b].
\end{equation}
Note that if the condition (\ref{eq:seperation_3}) or the collision-free condition 
(\ref{eq:seperation_1}) is satisfied for some feedback code at a fixed $k$,  
then these conditions remain satisfied if fewer than $k$ users are active. 

The rate of the feedback scheme is defined to be $R_f = \log(T)$,
when the encoder output is represented by a fixed-length code. The goal of this paper is to find minimum-rate encoding and decoding rules 
for the collision-free scheduling of the active users. To this end, we define $T^*(n,k,b)$ to be the minimum size of the encoder output alphabet such that there exists an encoder $f$ and a set of decoders $g_i$ that satisfy \eqref{eq:seperation_1} if $b \geq k$ or \eqref{eq:seperation_3} if $b < k$. Then, we define the minimum rate of a fixed-length collision-free feedback scheduling code as
\begin{equation}\label{eq:rf_star}
	R^*_f(n,k,b) \triangleq \log(T^*(n,k,b)).
\end{equation}

We can improve upon the above rate, if variable-length entropy coding is used on the encoder output. In this case, it is useful to consider a random activity pattern $\mathbf{A}$ defined over the sample space $\binom{[n]}{k}$ with probability distribution $Q(\mathbf{A})$, and further to recognize that the output of the encoding function $f(\mathbf{A})$ to ensure collision-free is not necessarily unique, so a judicious choice of output may reduce entropy. 
The variable-length feedback rate can be defined as $R_v = H(f(\mathbf{A}))$.
The rationale is that  entropy coding can be applied to the encoder output to achieve the above rate on average to within 1 bit. Our goal is to find the minimum-rate encoding and decoding rules 
for the collision-free scheduling of the active users. Note that the above rate depends on the distribution $Q (\cdot)$ as well as the choice of encoder $f$. Instead of considering a specific $Q (\cdot)$, in this work we choose to consider the worst-case $Q(\cdot)$. To this end, we define $f^*_Q$ to be an encoder that minimizes $H(f(\mathbf{A}))$ for a distribution $Q(\cdot)$ under the constraint that there exists a set of decoders $g_i$ that together with $f^*_Q$ satisfy \eqref{eq:seperation_1} if $b \geq k$ or \eqref{eq:seperation_3} if $b < k$. Then we define the minimum rate of a variable-length collision-free feedback scheduling code as
\begin{equation} \label{eq:rv_star}
 R^*_v(n,k,b)  \triangleq \sup_{Q(\cdot)} H(f^*_Q (\mathbf{A})).
\end{equation}
The two formulations \eqref{eq:rf_star} and \eqref{eq:rv_star} 
each have their utility. Feedback via a fixed-length code is easier to implement, 
but a variable-length code achieves a lower feedback rate.

The above formulation frames scheduling as  a zero-error problem, where no two users are ever assigned the same slot. As we show in the following sections, this formulation allows us to leverage existing results in combinatorics and develop concise bounds.  In the motivating massive random access problem, however, the scheduling phase is only one part of a three-phase scheme where the other two phases consist of transmissions over a noisy channel and generally have some small probability of error. Thus, non-zero end-to-end error probability is unavoidable. The study of the scheduling problem where some small probability of collision is allowed may be an interesting case for future study.

 In the first part of this paper, 
we focus on the $b=k$ case. When this is the case, we use the notation 
$R_f^*(n,k) \triangleq R_f^*(n,k,k)$ and $R_v^*(n,k) \triangleq R_v^*(n,k,k)$ 
for simplicity.  When $b > k$, the minimum rate to ensure
collision-free scheduling is clearly a non-increasing function of $b$ for 
fixed $k$ and $n$. In fact, having a large number of scheduling slots can
significantly reduce the feedback rate. This trade-off is investigated in
the second part of this paper. The third part of the paper investigates extension to the $b < k$ case.

\subsection{Naive Scheme is Suboptimal}

As mentioned earlier, a simple way to ensure collision-free scheduling is to
assign a unique index to each of $n$ users, then the feedback code simply
consists of listing the $k$ active users in the order they should
transmit. Each user finds its index in the list, waits for its turn, then
transmits at its scheduled slot. Thus, a feedback rate of $R = k \lceil \log(n)
\rceil$ is achievable. 

This naive scheme is not optimal for several reasons. 
Observe that the above simple scheme specifies a precise order in which $k$
users should transmit, but there are $k!$ collision-free schedules over the $k$
users.  It is therefore possible to use the flexibility of only having to specify
\emph{one} of the $k!$ schedules to reduce the feedback rate.
(If $b \geq k$ slots are available, then the number of collision-free schedules
increases to $\binom{b}{k}k!$.)  This is precisely what is done in
\cite{Facenda2020}, which uses enumerative source coding \cite{Cover1973} to reduce the feedback rate for collision-free scheduling.
In this case, the encoder associates one output symbol for each activity pattern,
resulting in a cost of feedback of $\log \binom{n}{k}$, which is a reduction of 
approximately $\log(k!)$ bits from the naive scheme. However, $\log \binom{n}{k}$
still scales as $O(\log(n))$ for fixed $k$.

This paper shows that significant further reduction in feedback rate is possible.
The key observation is that both the naive scheme and the enumerative source coding
scheme of \cite{Facenda2020} reveal the identities of all the active users and their scheduled slots to everyone. 
Such a feedback strategy clearly contains extraneous information from each user's perspective, as each user only needs to
know in which slot it should transmit and does not care about the schedules of the
other users. It turns out that, interestingly, even though a common feedback message needs to be designed for all users, it is possible to design an efficient feedback scheme to eliminate this extraneous information and to achieve
a rate that scales as $k \log(e)$, plus a term that scales in $n$ as $\Theta(\log \log(n))$ for fixed $k$ in the fixed-length case, or an $O(1)$ term in the variable-length case.

To put the relative savings in feedback rate in perspective, consider a
practical example with $n=10^6$ and $k=b=10^3$.  The cost of feedback using the
naive scheme is approximately $20$ kilobits. The enumerative source coding
scheme of \cite{Facenda2020} requires approximately $11.4$ kilobits. In contrast, 
the result of this paper shows that collision-free scheduling is theoretically possible with 
only about $1.46$ kilobits of feedback.

\subsection{Two-User Example}

To illustrate how to do significantly better than the $O(\log(n))$ scaling,
consider the following example of $k=b=2$ and $n \gg k$. 
The two active users are chosen uniformly at random among a larger
number of potential users.  The task is to schedule the two active users in 
two slots. The naive scheme requires $2\log(n)$ bits of feedback.

\subsubsection{Fixed-Length Code}
The following fixed-length code requires significantly less feedback.  
Index each of the $n$ users
with $\lceil \log(n) \rceil$ bits, using a binary representation of its index.
Since the binary representations of any two distinct non-negative integers must
differ in at least one position, we can use a feedback scheme that specifies the
location where the indices of the two users differ. Each user would examine the
bit value of its own index at that location. The user with bit value 0 would
transmit first; the user with bit value 1 would transmit second, thus
avoiding collision. Since specifying an index location of a vector of length $\lceil \log(n) \rceil$
requires $R_f = \lceil \log \lceil \log(n) \rceil \rceil$ bits, we achieve
$O(\log(\log(n)))$ scaling for collision-free feedback!

\subsubsection{Variable-Length Code} 

If we permit variable-length codewords, we can design a code with even lower
average rate. For simplicity, assume that $n=2^l$ for some $l$. 
Observe that in most cases the binary representation of two distinct
non-negative integers differ in many positions. This presents flexibility in
choosing the codeword.

To reduce the feedback rate, we need to choose an encoding rule that
minimizes the entropy of the encoder output. One such choice of an entropy 
minimizing encoder is the one that outputs the most significant position
where the indices of the two users differ. If the activity pattern is uniform
at random, it can be shown that the output of such an encoder is
distributed as a truncated and scaled geometric distribution. More precisely,
the probability distribution of the most significant position where two 
randomly chosen length-$l$ binary numbers differ is 
\begin{equation} \label{var_2user_pdf} 
p(t) = \frac{2^{t-1}}{\binom{n}{2}}
\left( \frac{n}{2^{t}}\right)^{2} = \frac{2^{l-t}}{2^{l}-1}, \;\; t \in \{ 1,
\dotsc , l \}.
\end{equation}
Encoding these symbols with a Huffman code, we can achieve an average rate of
\begin{equation} \label{eq:var_rate_2user} {R}_v \le 2 - \frac{\log(n)+1}{n-1}.
\end{equation}
With this scheme, the average rate is strictly less than a constant, even for 
large $n$. Indeed, the average rate $R_v \rightarrow 2$ as $n
\rightarrow \infty$.

In both cases, the key to achieving the saving in feedback rate is in assigning
multiple ``compatible" activity patterns to the same feedback output, then
defining decoding rules that result in zero collision for all ``compatible"
activity patterns.

\subsection{Reformulation via Set Partitioning}

The idea of defining ``compatible'' activity patterns can be generalized to
arbitrary $(n,k,b)$ and made rigorous using the following characterization
of an encoder and decoders. We restrict to the $b \ge k$ case here and defer
the $b < k$ case to Section \ref{sec:multiuser}. 
\begin{defn}
Define a $b$-partition of a set $\left[n\right]$ to be an ordered tuple of
subsets 
\begin{equation}
\part = (\mathbf{X}_1, \dotsc , \mathbf{X}_b)
\end{equation}
such that $\mathbf{X}_i \cap \mathbf{X}_j = \varnothing$, $\forall i \neq j$, and
$\bigcup\limits_{i=1}^{b}\mathbf{X}_i = \left[n\right]$.
\end{defn}
\begin{defn}
For fixed $\part$, define $\partcover{}$ as the following set of size-$k$ subsets of $[n]$:
\begin{equation}
\partcover{} =\\
 \left\{ \mathbf{A} \;\middle|\; \left \lvert \mathbf{A} \cap \mathbf{X}_{i}\right \rvert \leq 1, \mathbf{A} \in \binom{[n]}{k}, i=1,\dotsc, b \right\}.
\end{equation}
\end{defn}
Intuitively, these size-$k$ subsets of $[n]$ correspond to the activity patterns
for which each active user belongs to a distinct subset in the partition.  The
idea is that by specifying a $b$-partition $\part$, for all activity patterns in
$\mathbf{C} \left(\part\right)$, each active user can simply look at which
subset it belongs to in the partition, then schedule itself in the slot
corresponding to the index of that subset in a collision-free manner. 

\begin{defn}
Define a family of $b$-partitions of the set $[n]$ as an ordered collection of $T$ partitions \begin{equation}
\hashfamily = \left( \part^{(1)}, \dotsc, \part^{(T)}\right),
\end{equation}
\end{defn}
To make sure that all activity patterns are covered, we construct a family of $T$ partitions $\hashfamily$ such that
\begin{equation} \label{eq:no_err_condition} 
\bigcup \limits_{t=1}^{T} \mathbf{C} \left( \part^{(t)} \right) 
= \binom{\left[ n \right]}{k}.
\end{equation}
Then, whenever an activity pattern occurs, the BS only needs to specify a
partition that covers the activity pattern. Note that just as in the
$k=2$ example considered earlier, it is possible that multiple partitions cover
the same activity pattern. 

Fixing such a family of partitions $\hashfamily$, we now define an encoding
function which maps from the activity pattern to $[T]$ as:
\begin{equation}
f(\mathbf{A}) =  t \;\text{such that}\; \mathbf{A} \in 
\mathbf{C}\left( \part^{(t)} \right), \label{eq:encoding}
\end{equation}
and decoding functions which map from $[T]$ to the scheduling slots as:
\begin{equation}
g_i\left(t\right) = j \;\text{if}\; i \in \mathbf{X}_j^{(t)}, \text{where}\;\;
 \part^{(t)} = (\mathbf{X}_1^{(t)},...,\mathbf{X}_b^{(t)}). 
\label{eq:decoding}
\end{equation}

By \eqref{eq:no_err_condition}, for any arbitrary 
realization of $\mathbf{A}$, one can always find $t$ to satisfy the condition
 \eqref{eq:encoding}. Since at most one active user is in each subset of
 $\part^{(t)}$, \eqref{eq:decoding} guarantees that the schedule is
 collision-free.

The above set-partition view of scheduling with feedback
is completely general in the sense that any choice of deterministic
$\left(f',g'_i \right)$ that achieves zero collision for every activity pattern at a feedback rate $R$ can be written in this set-partition framework.
Let $T=\lceil 2^R \rceil $.  Given the decoding functions
$g'_i: [T] \to \left[ b\right]$,
we can define $T$ partitions $\part^{\left( t\right)} = \left(
\mathbf{X}^{\left( t\right)}_1, \dotsc , \mathbf{X}^{\left( t\right)}_b\right)$,
$t \in \left[ T\right]$, where 
\begin{equation}
\mathbf{X}^{\left(t\right)}_j = \left\{i \;\lvert \; g'_i \left( t \right) =
j, i \in \left[ n\right] \right\}.
\end{equation}
Using this construction,
partition $t$ covers precisely the activity patterns for which the feedback
symbol $t$ results in no collision. Since the set of functions $g'_i$ needs to
result in no collision for every activity pattern $\mathbf{A}$, this means that
\eqref{eq:no_err_condition} must be satisfied. Since \eqref{eq:no_err_condition}
is satisfied, and the code is collision-free, this implies that $f'$ is of
the form of \eqref{eq:encoding}.

Without loss of generality, we can therefore restrict attention to this
set-partition strategy for finding the minimum feedback rate for scheduling $k$
out of $n$ users in a collision-free manner in $b$ slots.
For the fixed-length code case, this means that the problem now reduces to finding the
minimum $T$ needed to satisfy \eqref{eq:no_err_condition}.  For the
variable-length code case, this means finding a family of partitions that satisfies
\eqref{eq:no_err_condition} and an encoding function \eqref{eq:encoding} so that 
$H \left( f \left(\mathbf{A}\right)\right)$ is minimized.

\subsection{Perfect Hashing and Graph Covering}

The minimum set-partition problem of the fixed length case with $b \geq k$ can be mapped to the
problem of finding an edge covering of a complete $k$-uniform hypergraph with a
set of complete $b$-partite subgraphs, and equivalently, the problem of finding a family
of perfect hashing functions. These connections allow us to leverage
existing results in combinatorics for tight bounds and for 
explicit feedback code constructions as discussed in Section \ref{sec:practical}.

A $k$-uniform hypergraph is a hypergraph where each hyperedge contains exactly $k$ vertices. Consider the $k$-uniform complete hypergraph $\mathcal{A} = \left( \mathbf{V},
\mathbf{E} \right)$ with $\mathbf{V} = \left[ n\right]$ and $\mathbf{E} =
\binom{\left[ n\right]}{k}$.  We can interpret the partition as a $b$-partite
complete subgraph of this hypergraph with edge
set $\partcover{}$.  Then, the question of whether every edge of a hypergraph
$\mathcal{A}$ is covered by a set of $T$ complete $b$-partite subgraphs can be
interpreted as precisely the condition \eqref{eq:no_err_condition}. Thus, finding a
family of complete $b$-partite subgraphs of $\mathcal{A}$ that covers $\mathcal{A}$
is equivalent to the minimum set-partition problem described earlier.  The
concept of graph entropy has been used to establish lower bounds on the minimum
$T^*$ required for edge covering \cite{Radhakrishnan1992ImprovedHypergraphs}.

The perfect hashing family problem goes back at least to \cite{melhorn}.
A fundamental work in this area is by Fredman and Koml{\'{o}}s
\cite{Fredman1984OnFunctions}. An $(n,b,k)$-family of perfect hash functions
is a family of functions from $[n] \to [b]$ for $n \geq b \geq k$ such
that for every $\mathbf{A} \subset[n], \left \lvert \mathbf{A} \right \lvert=k,$
there exists a function in the family that is injective on $\mathbf{A}$. An $(n,
k)$-family of minimal perfect hash functions is an $(n,b,k)$-family of perfect
hash functions where $b = k$.  We can view the decoding functions
\eqref{eq:decoder} as a family of $T$
functions from $[n] \to [b]$, if we swap the argument and the subscript.
With this interpretation, the decoding functions in this paper are exactly an $(n,b,k)$-family of minimal perfect hash functions. 
A fundamental result on perfect hash functions is:

\begin{thm}[Fredman and Koml{\'{o}}s \cite{Fredman1984OnFunctions}, 
	K\"orner \cite{Korner1986FredmanKomlosTheory},
	K\"orner and Marton \cite{Korner1988NewTheory}]
\label{thm:graph_entropy}
The minimal number of functions $T^{*}$ in an $\left(n,b,k \right)$-family of
perfect hash functions is bounded by: 
\begin{equation}\label{eq:perfect_hash}
\frac{\log n}{\min _{1 \leq s \leq k-1}  \frac{b^{\underline{s}}}{b^{s}} \log \frac{b-s+1}{k-s}} \lesssim T^{*} \lesssim \frac{\left(k -1\right) \log n}{\log \left( \frac{1}{1-\frac{b^{\underline{k}}}{b^{k}}} \right)}.
\end{equation}
\end{thm}

The notation $U(n) \lesssim V(n)$ means $U(n) \leq (1+ o(1)) V(n)$, where the
$o(1)$ term tends to zero as $n$ tends to infinity for fixed $b,k$.  The upper
bound in Theorem \ref{thm:graph_entropy} is derived using a random hashing argument.
The proof of the lower bound is more involved, and draws
heavily from graph theory. In particular, the proof uses the notion of the entropy of hypergraphs together with other information theoretic arguments. Nilli \cite{nilli1994perfecthash} later showed that it is possible to arrive at the same
result with a simpler probabilistic argument.

Although Theorem \ref{thm:graph_entropy} can already be directly applied to our minimum 
feedback problem for collision-free scheduling for the case of fixed-length codes 
with $b \geq k$, in order to keep this work self-contained and accessible
without requiring a background in combinatorics and graph theory, 
we present simpler versions of these bounds and their proofs for the
aforementioned cases in Sections \ref{sec:upper} and \ref{sec:lower}. For the
$b < k$ or variable-length cases, novel bounds and new proofs are presented.  To distinguish bounds
that can be derived from prior classic works in perfect hashing
from the new results, we state the former as propositions and the latter as theorems.

We remark
that perfect hashing families are also connected to a related, albeit different,
random-access problem \cite{Berger1981,Hajek1987,Korner1988RandomAccess}.
Additionally, both hashing and random access are often framed as classic
balls-into-bins problems \cite{Johnson1977}. For example, the throughput of the classic slotted ALOHA algorithm is related to the probability of collision in the balls-into-bins problem.
Similarly, balls-into-bins problems are also used in the analysis of hashing algorithms.  \section{Achievable Minimum Feedback Rate for $b = k$} \label{sec:upper}
We begin with the $b=k$ case and aim to find the minimum common feedback rate
required for collision-free scheduling of $k$ out of $n$ users into $k$ slots,
i.e., $R_{f}^{*}(n,k)$ for the case of fixed-length code and $R^{*}_{v}(n,k)$ for the case of variable-length code.

\subsection{Random Partition}

The main challenge in designing a feedback scheduling code is the explicit construction
of a family of partitions to cover all activity patterns.
In this section, we propose a random code construction.
The key idea is to bound the probability that a randomly chosen family of $T$
$k$-partitions of $[n]$ satisfies the collision-free condition
\eqref{eq:no_err_condition}. 
For the case of fixed-length codes, if for some integer $T$, this probability is 
strictly greater than zero, this would imply the existence of one collision-free 
feedback code at a rate $R_f=\log(T)$, hence an upper bound on $R^{*}_{f}(n,k)$. 
For the case of variable-length code, we consider a greedy encoder 
that satisfies \eqref{eq:no_err_condition}
and upper bound the entropy of the output of the encoder, 
thus obtaining an upper bound on $R^{*}_{v}(n,k)$. The following bound is useful for both the fixed and variable-length code cases.

\begin{lem}\label{lem:achieve_epsilon}
Let $\part^{(1)}, \dots, \part^{(T)}$ be a family of $k$-partitions of the set
$[n]$, where each partition is uniformly and independently chosen over the
set of all $k$-partitions of $[n]$ with $k < n$. Let $T_{\epsilon}^*(n,k)$ be the smallest 
integer such that the probability that this
family satisfies condition \eqref{eq:no_err_condition} is at least $1-
\epsilon$ with $0 < \epsilon < 1$. Then, $T_{\epsilon}^{*}(n,k)$ is upper bounded
by
\begin{equation}  \label{eq:achieve_epsilon}
T_{\epsilon}^{*}\left(
n,k\right) \leq \ln{\left(\frac{1}{\epsilon} \binom{n}{k}\right)
\left( \frac{k^k}{k! }\right)}.
\end{equation}
\end{lem}
\begin{IEEEproof}
Random $k$-partitions of $[n]$ uniformly distributed over the set
of all partitions can be constructed in the following way. 
For each element in $[n]$, we draw a uniform random variable in $[k]$ to 
denote which subset this element belongs to. If $\part{}$ is constructed 
randomly this way, then fix some $\mathbf{A} \in \binom{\left[ n\right]}{k}$, 
the probability that $\mathbf{A}$ is covered by $\partcover{}$, i.e., 
each of the $k$ elements in $\mathbf{A}$ lies in a distinct subset of $\part^{}$, 
is $\frac{k}{k}\times\frac{(k-1)}{k} \times \cdots \times \frac{1}{k}$. This is to say that 
\begin{equation}
\label{eq:prob_random_covering}
\mathrm{Pr}\left( \mathbf{A} \in  \partcover{} \right) = \frac{k!}{k^k} \triangleq p.
\end{equation}
If $T$ random $k$-partitions are generated independently, the probability that
none of the $T$ partitions cover this given activity pattern is
\begin{equation} \label{eq:prob_uncovered}
\mathrm{Pr} \left( \mathbf{A} \notin \bigcup
\limits_{t=1}^{T} \partcover{(t)} \right) = \left(1 - p\right)^T.
\end{equation}
By applying the union bound, we find that the probability that all $\mathbf{A} \in \binom{\left[ n \right]}{k}$ are covered by the $T$ partitions is bounded as:
\begin{equation}\label{eq:union_bound_lem1}
\mathrm{Pr} \left(  \binom{\left[ n \right]}{k} \neq \bigcup
\limits_{t=1}^{T} \partcover{(t)} \right) \leq \binom{n}{k}\left(1 - p\right)^T .
\end{equation}

We require the above probability to be less than $\epsilon$. Using the fact that $(1 - x) < e^{-x}, \forall x > 0$, this
gives us the following sufficient condition on $T$ that ensures
\eqref{eq:no_err_condition} is satisfied by a random family of $T$ partitions with probability at least $1 - \epsilon$:
\begin{equation}
\binom{n}{k} \exp \left(-{pT}\right) \leq \label{eq:lt_eps}
\epsilon.
\end{equation}
Taking logarithm of both sides and substituting $p=\frac{k!}{k^k}$ yields 
\begin{equation}\label{eq:lemma1_final}
T \geq \ln \left( \frac{1}{\epsilon}\binom{n}{k}
\right)\left(\frac{k^k}{k!}\right).
\end{equation}
Hence, $T_{\epsilon}^*(n,k)$ must be bounded above as in \eqref{eq:achieve_epsilon}.
Note that for $1 \leq k < n$, the right-hand side of the above expression is always
greater than 1. \end{IEEEproof}

\subsection{Fixed-Length Code}

We now provide an upper bound on the minimum rate 
of a collision-free fixed-length feedback scheduling code  $R^*_{f}(n,k)$.

\begin{prop}\label{thm:fixed_achieve_bound}
For $k < n$, the minimum rate of a fixed-length collision-free feedback scheduling code $R^*_{f}(n,k)$ is upper bounded as
\begin{equation}\label{eq:above_bound}
R_{f}^*(n,k) \leq
k\log(e) +\log \left( \ln \left(\frac{n}{k}\right) +1 \right) +
\frac{1}{2}\log{\left(\frac{k}{{2\pi}} \right)}.
\end{equation}
Thus, for massive random access, there exists a fixed-length feedback 
code for scheduling $k$ out of $n$ users in $k$ slots with no collision at
the above rate, which scales as $\log(e)$ bits per active user,  
plus a term that scales in $n$ as $O(\log(\log(n)))$ for fixed $k$. 
\end{prop}

\begin{IEEEproof}
Since by Lemma \ref{lem:achieve_epsilon} for any choice of $\epsilon$
with $0 < \epsilon < 1$, a randomly chosen family of $T_{\epsilon}^{*}(n,k)$ 
partitions would satisfy \eqref{eq:no_err_condition} with a non-zero
probability of $1-\epsilon$, this implies the existence of at least one family of partitions with
rate $\log \left(T_{\epsilon}^{*}\left(n,k\right)\right)$ that results in collision-free scheduling.
In other words, the rate
\begin{equation} 
R = \log \left(  \ln \left( \frac{1}{\epsilon} \binom{n}{k} \right)
	\left(\frac{k^k}{k!}\right)\right) 
\end{equation}
is achievable for any $0 < \epsilon < 1$. Take the minimum of the above by letting $\epsilon \rightarrow 1$. This means that any rate
\begin{equation} \label{eq:fixed_rate_inf}
R > \log \left(  \ln \left( \binom{n}{k} \right)\left(\frac{k^k}{k!}\right)\right) 
\end{equation}
is achievable. 
Noting that $\binom{n}{k} < \frac{n^k}{k!}$, \eqref{eq:fixed_rate_inf} ensures that if 
\begin{equation} 
R \geq k\log(k) -\log(k!) + \log\left(\ln{
\frac{n^k}{k!}}\right), 
\end{equation} 
then there must
exist a collision-free feedback code for scheduling $k$ out of $n$ users.
Using the fact that $k! > \sqrt{2 \pi} k^{k + \frac{1}{2}} e^{-k}
e^{\frac{1}{12k + 1}}$, we arrive at the following sufficient condition on $R$
for the existence of a collision-free fixed-length feedback scheduling code:
\begin{equation}\label{eq:prop1_ratebound}
 R \geq k \log \left( e \right) +  \log\left( \ln \left(
\frac{n}{k}\right) +1 \right) + \frac{1}{2} \log{\left(\frac{k}{2 \pi}\right)}.
\end{equation}
 The achievability of the above rate means that the minimum rate
$R^*_{f}(n,k)$ must be upper bounded by the right-hand-side of \eqref{eq:prop1_ratebound}, which gives (\ref{eq:above_bound}).
Thus the minimum rate scales at most as $\log(e)$ bits per active user, plus an
additive term that scales in $n$ as $O \left(\log \log \left( n \right)\right)$ for fixed $k$.
\end{IEEEproof}

The minimum rate collision-free feedback scheduling codes are closely related to
perfect hash families. The argument of the proof above is similar to the arguments
used for proving bounds in the perfect hash function literature
\cite{Fredman1984OnFunctions}. 
The key observation here is that for a fixed-length code, this random 
coding bound results in an achievable rate that has an $O\left( \log \log (n)\right)$ scaling in $n$ for fixed $k$. 
We next show that
if a variable-length code is used, the average rate can be further 
reduced so as to remove even this $\log \log (n)$ dependence. 

\subsection{Variable-Length Code}

We now upper bound the minimum rate of collision-free variable-length feedback scheduling
code $R_v^*(n,k)$.  To bound the achievable rate for the
fixed-length code, only the value of $T$ is important. However, for the 
variable-length code case, the important quantity is 
$H\left(f \left( \mathbf{A}\right)\right)$, where $\mathbf{A}$ 
is assumed to follow some distribution, and the output of $f(\mathbf{A})$ is judiciously chosen. More specifically, since
the encoder typically has flexibility in choosing which
partition to use to cover the activity pattern, we also must fix a
particular choice of encoding function in order to characterize the entropy of
the encoder output. For this purpose, we define the following ``greedy
encoder", which among all encoding functions that satisfy
\eqref{eq:encoding} tends to produce a highly skewed output
distribution, hence giving a lower output entropy.

\begin{defn}
Given a family of $T$ $k$-partitions $\hashfamily = (\part^{(1)}, \dotsc, \part^{(T)})$, define the greedy encoder $f_{\hashfamily}$:
\begin{equation}\label{eq:greedy_encoder}
\begin{aligned}
f_{\hashfamily}(\mathbf{A}) = \\ \;\;  
\end{aligned}\;\;\;\;\; 
\begin{aligned} \min_{t \in [T]} \quad
& t\\ \mathrm{s.t.} \quad & \mathbf{A} \in \partcover{(t)}.
\end{aligned}
\end{equation}
In the case where not all activity patterns are covered by the family of 
partitions $\hashfamily$, i.e., \eqref{eq:no_err_condition} is not satisfied,
the minimization problem may be infeasible for some $\mathbf{A}$. In this case, we take $f_{\hashfamily}(\mathbf{A}) = T+1$.
\end{defn} 

To characterize the output entropy of $f_{\hashfamily}(\mathbf{A})$ for 
a particular $\hashfamily$ is not easy. Instead, we study the behavior 
of $f_{\hashfamily}(\mathbf{A})$ over all possible $\hashfamily$'s in
order to show the existence of a good collision-free variable-length code. 
The main result of this section is the following. 

\begin{thm}\label{thm:achieve_source_coded}
For $k \leq n$, the minimum rate of a collision-free variable-rate feedback scheduling code $R^{*}_{v}(n,k)$ is upper bounded as
\begin{equation}\label{eq:var_rate_upper}
R^{*}_{v}(n,k) \le \left( k  +1\right)\log \left( e \right).
\end{equation}
Thus for massive random access, there exists a
variable-length feedback code for scheduling $k$ out of $n$ users into
$k$ slots with no collision and an average rate as
expressed in \eqref{eq:var_rate_upper}, which scales as $\log(e)$ bits per active user
and is independent of $n$.
\end{thm}

\begin{IEEEproof}
The main idea is to consider the average output entropy of the greedy encoder over
all families of partitions, regardless of whether they satisfy the covering
property \eqref{eq:no_err_condition}.  Then, we separate all families of
partitions into two sets: one over families of partitions that satisfy the
covering property \eqref{eq:no_err_condition}, and the other over families that
do not.  By accounting for the contribution to the average entropy of the
families that do not satisfy \eqref{eq:no_err_condition}, it is possible to
bound the average output entropy of the families that do, which in turn allows
us to establish the existence of one family of partitions and a corresponding 
encoder with the required rate and the zero-collision property.

Toward this end, we denote the set of all families of $k$-partitions of $[n]$
of size $T$ as $\hashfamilyset (n,k,T)$.  Fix $n$ and $k$. We choose some
$\epsilon \in (0,1)$ and set $T$ to be
$T_{\epsilon}^*(n,k)$ as defined in Lemma \ref{lem:achieve_epsilon}. 
We aim to study the output of the greedy encoder $f_{\hashfamily}$ where
$\hashfamily \in \hashfamilyset \left(n,k,T_{\epsilon}^{*}\left(n,k\right)\right)$. 
For the rest of the proof, we shall refer to 
$\hashfamilyset \left(n,k,T_{\epsilon}^{*}\left(n,k\right)\right)$ simply as 
$\hashfamilyset$ for brevity. 

Specifically, for each $\hashfamily \in \hashfamilyset$,  
for activity patterns $\mathbf{A}$ distributed according to $Q(\mathbf{A})$ and 
under the greedy encoder, define
\begin{equation}
p_{\hashfamily}(t) =  \mathrm{Pr}\left( f_{\hashfamily} (\mathbf{A}\right) = t ), \quad t=1,\cdots,T.
\end{equation}
Note that since many of these families of partitions in $\hashfamilyset$
do not cover all $\mathbf{A}$, we define
\begin{equation}
p_{\hashfamily}(T+1) =  1- \sum_{t=1}^T p_{\hashfamily}(t)
\end{equation}
in order to make $p_{\hashfamily}(t)$ a probability distribution over $\{ 1,\cdots,T, T+1 \}$. 

Now, let $\hashfamily$ be chosen at random from a uniform distribution over all
possible families of partitions $\mathcal{U}\left( \hashfamilyset\right)$. 
For such a random partition, define
\begin{equation}
p_{\randhashfamily}(t) \triangleq \mathbb{E}_{\randhashfamily} \left[ p_{\hashfamily}(t)\right].
\end{equation}
Since the entropy
function is concave, Jensen's inequality implies
\begin{equation}
H  ( p_{\randhashfamily} ) = H
\left( \mathbb{E}_{\randhashfamily} \left[p_{\hashfamily} \right] \right) \geq
\mathbb{E}_{\randhashfamily}\left[ H(p_{\hashfamily})\right].
\end{equation}
Let $\hashfamilyset_1$ represent the set of all $\hashfamily \in \hashfamilyset$
that satisfy \eqref{eq:no_err_condition}, while $\hashfamilyset_2$ represents
all $\hashfamily \in \hashfamilyset$ that do not satisfy
\eqref{eq:no_err_condition}. We can decompose the expectation over $\randhashfamily
\sim \mathcal{U}\left(\hashfamilyset \right)$ into an expectation over
$\randhashfamily_1 \sim \mathcal{U} \left( \hashfamilyset_1 \right)$ and
$\randhashfamily_2 \sim \mathcal{U}\left( \hashfamilyset_2 \right)$. If $\alpha$ is the fraction of
$\hashfamily \in \hashfamilyset$ that satisfy \eqref{eq:no_err_condition}, we have
\begin{IEEEeqnarray}{RCl}
H(p_{\randhashfamily})& \geq &
\alpha \mathbb{E}_{\randhashfamily_1} \left[ H  (p_{\hashfamily_1} )\right] 
+ (1 - \alpha)\mathbb{E}_{\randhashfamily_2} \left[H ( p_{\hashfamily_2} )\right]\\
& \geq & \alpha \mathbb{E}_{\randhashfamily_1} \left[ H
(p_{\hashfamily_1})\right]. 
\end{IEEEeqnarray}
Now, we use Lemma \ref{lem:achieve_epsilon} to lower bound $\alpha$. Since $\hashfamilyset$  is chosen such that there are $T_{\epsilon}^*(n,k)$ partitions in each family, Lemma \ref{lem:achieve_epsilon} says that at least a fraction of $(1 -\epsilon)$ of the families in $\hashfamilyset$ satisfy \eqref{eq:no_err_condition}. Therefore, $\alpha > 1-\epsilon$ and
\begin{equation}
H(p_{\randhashfamily}) \geq (1 - \epsilon) \mathbb{E}_{\randhashfamily_1} \left[ H (p_{\hashfamily_1})\right].
\end{equation}
Since the minimum entropy $H ( p_{\hashfamily_1})$ for $\hashfamily_1 \in \hashfamilyset_1$ cannot exceed its average $\mathbb{E}_{\randhashfamily_1}\left[ H(p_{\hashfamily_1})\right]$, we must have
\begin{equation}
\label{eq:min_hash_entropy}
\min_{\hashfamily_1 \in \randhashfamily_1} H ( p_{\hashfamily_1}) \leq \frac{1}{1 - \epsilon}H(p_{\randhashfamily}).
\end{equation} 
Thus, if we can bound $H(p_{\randhashfamily})$, this would give us a bound on 
the output entropy of the family of partitions with the lowest rate that satisfies the zero-collision condition (\ref{eq:no_err_condition}). 

It remains to bound $H(p_{\randhashfamily})$. The first
partition $\part^{(1)}$ of a random $\hashfamily \in \randhashfamily$ is equally likely to be
\textit{any} $k$-partition of $[n]$. Thus, for an $\mathbf{A}$ which
consists of $k$ elements of $[n]$, the probability that each of the $k$
elements lies in a distinct subset of $\part^{(1)}$ is as computed in (\ref{eq:prob_random_covering}), i.e., \begin{equation}
p_{\randhashfamily}(1) = \mathrm{Pr}\left(\mathbf{A} \in \partcover{(1)}\right) = \frac{k!}{k^k}.
\end{equation}
Since we employ a greedy encoder, for $t=2,\dots, T$, we have
\begin{multline}\label{eq:greedy_enc_unsimplified}
\mathrm{Pr}\left( f_{\randhashfamily} \left(\mathbf{A}\right) = t\right) \\
= \mathrm{Pr}\left(\mathbf{A} \in \partcover{(t)}\right) \prod \limits_{\tau=1}^{t-1}
    \mathrm{Pr}\left(\mathbf{A} \notin \partcover{(\tau)}\right),
\end{multline}
but since each partition in $\hashfamily \sim \mathcal{U}(\randhashfamily)$ is i.i.d., this simplifies to
\begin{equation}\label{greedy_expect}
p_{\randhashfamily}(t) = \frac{k!}{k^k}\left( 1 -
\frac{k!}{k^k}\right)^{t-1}, \quad t=1, \dotsc, T.
\end{equation}
Note that
$p_{\randhashfamily}(t)$ matches the geometric distribution with parameter
$\frac{k!}{k^k}$ for $t = 1, \cdots, T$, and has the remainder of the mass at $t = T+1$.
The entropy of the geometric distribution
serves as an upper bound to the entropy of $p_{\randhashfamily}(t)$, i.e.,
\begin{IEEEeqnarray}{rCl}
 H(p_{\randhashfamily}) & = & -\sum_{t = 1}^{T+1} p_{\randhashfamily}(t) \log\left(p_{\randhashfamily}(t)\right) \\
&\leq & -\left(\log
    \left( \frac{k!}{k^k}\right) + \left( \frac{k^k}{k!} - 1 \right)\log \left(
    1 - \frac{k!}{k^k}\right)\right). \label{eq:geometric_limit}
\end{IEEEeqnarray}
Note that the above bound is independent of $Q(\mathbf{A})$, so by \eqref{eq:min_hash_entropy}
the following rate is achievable for all $Q(\mathbf{A})$:
\begin{equation}\label{eq:unsimplified_var_achieve}
R =  \frac{1}{1 - \epsilon}\left(\log \left( \frac{k^k}{k!}\right) 
+ \left( 1 - \frac{k^k}{k!} \right)\log \left( 1 - \frac{k!}{k^k}\right)\right).
\end{equation}
Noting that $k! > k^k e^{-k}$ and
\begin{equation}
	\label{eq:p_inequality_loge}
\left(1- \frac{1}{p} \right)\log \left(1 - p \right) < \log(e), \quad \forall\; 0 < p < 1,
\end{equation} 
this shows
\begin{equation} \label{eq:simplified_var_achieve}
R = \frac{1}{1 - \epsilon} \left( k    +  1\right)\log (e) 
\end{equation}
is achievable.
Since $\epsilon$ can be chosen arbitrarily, we can take $\epsilon \rightarrow 0$. 
As the minimum feedback rate must be no larger than any achievable rate, and this rate is achievable for any distribution of $\mathbf{A}$, this gives 
(\ref{eq:var_rate_upper}).
\end{IEEEproof}

To summarize, this section shows that while for a fixed-length code,
the random coding bound results in $O(\log \log(n))$ scaling for the
minimum rate, for the variable-length codes, the minimum rate is
independent of $n$.  This is surprising as it implies that even as $n$
grows to infinity, there exists an $O(k)$ feedback scheme that results
in no collision. Interestingly, the proof of the variable-length code case
shows that the way to achieve this is to take $\epsilon \rightarrow
0$, which implies $T_{\epsilon}^*(n,k) \rightarrow \infty$.  In other
words, to minimize the output entropy, it is better to use families of
large size $T$ so as to give more flexibility to the encoder in order
to produce a more skewed output distribution. This is in contrast to
the fixed-length case where $T$ is to be minimized and $\epsilon$
should be taken to 1.

 \section{Converse for Minimum Feedback Rate for $b = k$} \label{sec:lower}
We now present converse results showing that a feedback code for scheduling $k$
out of $n$ users in $k$ slots must have a minimum rate with at least a linear
scaling in $k$ for both the
variable and fixed-length case and a double-logarithmic scaling in $n$ for the fixed-length
case. The first result is based on the so-called volume bound, which is known in, e.g.,
\cite{Fredman1984OnFunctions} for the case of fixed-length codes. 
This argument can be extended to the variable-length case by considering 
the volume bound as a constraint on the distribution of $f(\mathbf{A})$ 
as $\mathbf{A} \sim Q(\mathbf{A})$ and over possibly random $f$.

\begin{thm}\label{thm:var_rate_vol}
For $k \leq n$, the minimum rate of a variable-length collision-free feedback scheduling code is bounded below as
\begin{equation}
	\label{eq:var_volume_bound}
R^{*}_{v}(n,k) \geq k \log(e) - \frac{1}{2} \log\left( 2 \pi k\right) - \frac{\log(e)}{12k}
- \log \left( \frac{n^{k}}{n^{\underline{k}}}\right).
\end{equation}
The same bound also applies to the fixed-length code case, i.e., to $R^{*}_{f}(n,k)$.
Thus, with either fixed or variable-length codes, scheduling $k$ random users 
out of a total of $n$ users into $k$ slots
in a massive random access scenario with no collision requires an average rate
that scales at least as $\log(e)$ bits per active user in the regime where $k\ll n$.   
\end{thm}

\begin{IEEEproof}
The number of activity patterns covered by a partition is maximized when the
sizes of the subsets of the partition take integer values surrounding $\frac{n}{k}$.
In particular, we can show that 
\begin{equation}
\left \lvert \partcover{\left( t\right)} \right \lvert \leq
\left \lceil{\frac{n}{k}}\right \rceil^{n\; \text{mod}\; k} \left
	\lfloor{\frac{n}{k}}\right \rfloor^{k- (n\; \text{mod}\; k)} \leq 
\left( \frac{n}{k}\right)^{k} \triangleq c_{\max}.
\label{eq:c_max}
\end{equation} 
Thus for the fixed-length code case, in order to cover all the activity patterns, 
i.e., to satisfy condition \eqref{eq:no_err_condition}, we must have 
\begin{equation}
T^{*}(n,k) \geq \frac{\binom{n}{k}}{\left( \frac{n}{k}\right)^{k}}.
\end{equation}
This bound is not necessarily tight, because the covering sets 
$\partcover{\left(  t \right)}$ are not necessarily disjoint, but it already 
provides the desired linear scaling bound for the fixed-length case. 
If we use the upper bound 
$k ! < \sqrt{2 \pi} k^{k+\frac{1}{2}} e^{-k} e^{\frac{1}{12 k}}$, we get
the right-hand side of \eqref{eq:var_volume_bound} for 
the fixed-length code case, i.e., for $R^{*}_{f}(n,k) = \log(T^*(n,k))$.

For the variable-length code case, consider an encoder that produces outputs 
according to \eqref{eq:encoding} for a family of partitions that satisfy
\eqref{eq:no_err_condition}. Define $v(n,k)$ to be the maximum probability 
that any such encoder $f$ could produce any given output $t$, i.e.,
let $v(n,k) = \max \limits_{f,t} \mathrm{Pr}\left(f ( \mathbf{A} )
= t\right)$ as $\mathbf{A} \sim Q(\mathbf{A})$. Since according to (\ref{eq:c_max}) the maximum number of 
activity patterns covered by any singe partition is bounded by 
$c_{\max} = \left( \frac{n}{k} \right)^k$, if we
order the fixed activity patterns $\mathbf{A}_i \in \binom{[n]}{k}$ such that
$Q(\mathbf{A}_1) \geq Q(\mathbf{A}_2) \geq \dots$, then we have
\begin{equation}\label{eq:volume_def}
v(n,k) \leq \sum_{i \leq c_{\max}} Q(\mathbf{A}_i). \end{equation}
Note that the above bound holds even if the encoding process is non-deterministic. 

Consider now the optimum variable-length encoder. 
Let $p^*(t) = \mathrm{Pr}\left( f^*(\mathbf{A}) = t\right)$. 
Since $v(n,k)$ is an upper bound on the probability that an activity pattern is mapped to $t$, we must have $p^*(t) \leq v(n,k) \; \forall t$. Therefore,
\begin{multline}\label{eq:log_p_vol_bound}
H ( p^*(t)  )  =  - \sum_t p^*(t) \log \left( p^*(t) \right) \\
 \geq  - \sum_t p^*(t) \log \left( v \left( n,k\right) \right) 
 =  - \log \left( v \left( n,k\right) \right). 
\end{multline}
Taking supremum of both sides yields
\begin{IEEEeqnarray}{rCl}
\sup_{Q(\cdot)}H ( p^* (t) )  & \geq & - \inf_{Q(\cdot)} \log(v(n,k)) \\
& \geq & - \log\left( \inf_{Q(\cdot)} \sum_{i \leq c_{\max}} Q(\mathbf{A}_i) \right) \\
\label{eq:var_volume_bound_unsimplified_3}
& = & - \log\left( \frac{c_{\max}}{\binom{n}{k}} \right) \\
& = & \log \left( \frac{k^k}{k!}\right) - \log \left( \frac{n^k}{n^{\underline{k}}}\right),
\label{eq:var_volume_bound_unsimplified}
\end{IEEEeqnarray}
where \eqref{eq:var_volume_bound_unsimplified_3} follows from the fact that the infimum is achieved by $Q\left( \mathbf{A}_i\right) = \binom{n}{k}^{-1}$ $\forall i\in \binom{[n]}{k}$, since any other choice would only increase the sum in \eqref{eq:volume_def}. Finally, using $k ! < \sqrt{2 \pi} k^{k+\frac{1}{2}} e^{-k} e^{\frac{1}{12 k}}$ gives
(\ref{eq:var_volume_bound}) as desired.

Note that the last term in \eqref{eq:var_volume_bound} is close to zero in the regime of interest (i.e., $k \ll n$). Thus, the minimum feedback rate must scale at least linearly in $k$ in this regime.
\end{IEEEproof}

With this result, we now have both upper and lower bounds on
$R^{*}_{v}(n,k)$. It is instructive to examine these bounds in the limit 
as $n \rightarrow \infty$ for fixed $k$. The unsimplified forms of the upper bound 
\eqref{eq:unsimplified_var_achieve} combined with (\ref{eq:p_inequality_loge}) and the lower bound
\eqref{eq:var_volume_bound_unsimplified} with $\frac{n^k}{n^{\underline{k}}} \rightarrow 1$ as $n \rightarrow \infty$ reveal that
\begin{equation}
\log \left( \frac{k^k}{k!}\right) + \log(e) \geq 
\lim_{n \rightarrow \infty} R^{*}_{v}(n,k) 
\geq \log\left( \frac{k^k}{k!}\right).
\end{equation}
Both bounds scale as $k\log(e)$ and the difference between the two bounds is an extra $\log(e)$ term, placing a fairly tight bound on the minimal average feedback rate.

Finally, we present a lower bound on $R^{*}_{f}(n,k)$ which says that
the rate must have at least an $O(\log\log(n))$ scaling.
\begin{prop} \label{thm:fixed_exclusion} 
For $1 < k \leq n$, the minimum rate of a fixed-length collision-free feedback scheduling code $R_{f}^*(n,k)$ is bounded below as
\begin{equation}\label{lower_2}
R^{*}_{f}(n,k) \geq  \log\log \left(\frac{n}{k-1}\right) + \log(k)-1. 
\end{equation}
Thus, in the fixed-length code case, scheduling $k$ out of a total of $n$ users over $k$ slots in a massive
random access scenario with no collision requires a feedback rate that scales
at least double logarithmically in $n$ for fixed $k$.
\end{prop} 
\begin{IEEEproof}
Consider the first partition.  We seek to bound the number of activity
patterns that this first partition cannot cover by noting that
\begin{equation}
\partcover{\left( 1 \right)} \bigcap \binom{\left[
n\right] - \mathbf{X}^{(1)}_j}{k} = \varnothing,  \;\; j = 1,
\dotsc, k.
\end{equation}
This is to say that $\part^{\left( 1 \right)}$ cannot
cover an activity pattern which has all its elements drawn from $[n]
- \mathbf{X}^{(1)}_j$. Since the partition must have at least one
subset $\mathbf{X}^{(t)}_j$ of size at most $\left
\lfloor{\frac{n}{k}}\right \rfloor$, it must be that $\part^{\left(
1 \right)}$ cannot have covered any activity patterns whose elements
are exclusively drawn from a set of indices of size $m_1 \left( n,
k\right)$, where 
\begin{equation}\label{eq:structure_bound}
 m_1(n, k) \geq n - \left \lfloor{\frac{n}{k}}\right \rfloor \geq n \left( 1 -
\frac{1}{k}\right).  
\end{equation} 
Now take the second partition, and consider how many of the above
activity patterns are still not covered by the second partition. By
the same logic, since the second partition cannot cover any activity
patterns drawn from an index set with indices from one of the
subsets of the partition removed, and when restricted to the set of
indices of size $m_1(n,k)$, there is at least one subset which
overlaps with at most $\left(1-\frac{1}{k}\right)$ portion of
$m_1(n,k)$ indices, we conclude that all the activity patterns whose
elements are drawn from an index set of size $m_2(n,k)$, where
\begin{equation}
 m_2(n, k) \geq n \left( 1 - \frac{1}{k}\right)^2,
\end{equation}
cannot possibly be covered by either the first
partition or the second partition.  Continuing for $T$ partitions,
the only way that the remaining indices cannot support any activity
patterns is for $m_{T}(n,k) \leq k-1$. This gives us the following
necessary condition on $T$: 
\begin{equation}
 n \left( 1 - \frac{1}{k}\right)^{T} \leq k-1.
\end{equation}
Since $T^*(n,k)$ must be
greater than or equal to any $T$ that satisfies the above, by taking
the logarithm of the above, we have 
\begin{equation} 
T^{*}(n,k) \geq \frac{\log(n) - \log(k-1)}{\log(k) - \log(k-1)}.
\label{Snir_bound}
\end{equation}
 In terms of rate, by taking
the logarithm again and by noting that  $-\log \left( 1 -
\frac{1}{k}\right) \leq \frac{2}{k}$ for $k > 1$, we get the
desired result, 
\begin{equation}
 R^{*}_{f}(n,k) \geq  \log\log \left(\frac{n}{k-1}\right) + \log (k)-1. \end{equation}
Thus for fixed $k$, the minimum feedback rate for
zero collision must scale at least double logarithmically in $n$.
\end{IEEEproof}
Note that in the context of hypergraph covering,
\eqref{Snir_bound} can be interpreted as Snir's bound
\cite{Snir1979TheHypergraphs}.
 \section{Feedback Rate for $b>k$} \label{sec:tradeoff}

We have so far discussed the case where the number of scheduling slots
$b$ is equal to the number of users $k$, but the case of $b > k$ 
is also of interest. For example,
communication networks must be designed for peak load, so there may be 
time periods throughout the day when the number of users is reduced but the number of
available slots remain fixed, leading to a situation where $b > k$.
When $b>k$, collision-free
scheduling becomes easier, so we would expect the minimum required
feedback rate to be lower as compared to the $b=k$ case.  Indeed when
the number of transmission slots is equal to $n$, each of the $n$ potential
users can be assigned a unique slot, so no feedback is required to
prevent collision. Larger $b$, of course, leads to waste of
transmission resources. Smaller $b$, such as $b=k$, wastes no slots, 
but requires larger feedback rate in order to avoid collision.  In this section, we study
the trade-off between $b$ and the minimum feedback rate, when $k<b<n$.

To characterize the minimum feedback rate to ensure collision-free
scheduling for the $b>k$ case, we extend the results of the previous 
sections. These extensions are straightforward, with the exception of Proposition
\ref{thm:fixed_exclusion}, where the previous proof relies on the fact that 
there are exactly $k$ slots. This case is however already covered in Theorem 
\ref{thm:graph_entropy}, which contains 
the $\log \log (n)$ scaling result for the case of $b \geq k$. The
proofs of the extensions of the other bounds are included in the appendix.
Note that these bounds are presented in
simplified forms which are similar to the bounds of the previous sections,
but due to the approximations used, they do not coincide with previous 
results in the $b=k$ case.

\begin{prop}\label{thm:trade_off_fixed_rate_achievable}
For $k < b \leq n$, the minimum rate of a fixed-length collision-free feedback scheduling code $R^*_{f}(n,k,b)$ is upper bounded as 
\begin{multline}\label{eq:achievable_fixed_tradeoff}
R^{*}_{f}(n,k,b) \leq k \log(e) + \log \left(\ln \left(\frac{n}{k} \right)+1\right) \\ + (b-k)\log \left(1-\frac{k}{b}\right) 
+ \log(k) + 1.
\end{multline}
\end{prop}

\begin{thm} \label{thm:trade_off_var_rate_achievable}
For $k < b \leq n$, the minimum rate of a variable-length collision-free feedback scheduling code $R^*_{v}(n,k,b)$ is upper bounded as
 \begin{equation}\label{eq:trade_off_var_rate_achievable}
R^{*}_{v}(n,k,b) \leq
 (k + 1) \log(e) + (b-k) \log \left( 1 - \frac{k}{b}\right) + 1.
\end{equation} 
\end{thm}

\begin{thm}\label{thm:trade_off_var_rate_volume}
For $k < b \leq n$, the minimum rate of a variable-length collision-free feedback scheduling code $R^*_{v}(n,k,b)$ is lower bounded as
\begin{multline}
R^{*}_{v}(n,k,b) \geq \\
k \log(e) 
+ \left( b-k + \frac{1}{2}\right)\log \left(1 - \frac{k}{b} \right)
 - \log \left( \frac{n^k}{n^{\underline{k}}}\right).
\end{multline}
The same bound also applies to the fixed-length code case, i.e., to $R^{*}_{f}(n,k,b)$.
\end{thm}

Comparing these bounds to the bounds for 
the case of $b = k$,
they all still have the same leading $k \log(e)$ term,
but these new bounds also have a term similar to $(b-k)\log \left( 1 - \frac{k}{b}\right)$ which decreases the required feedback rate for $b > k$. 
Using the achievable rate of the fixed-length case \eqref{eq:achievable_fixed_tradeoff}
as an example, suppose $b = \beta k$ for some $\beta > 1$, the reduction in feedback rate
is 
\begin{equation}
(b-k)\log \left( 1 - \frac{k}{b}\right) 
= - k \cdot \underbrace{(\beta - 1) \log \left( \frac{\beta}{\beta-1} \right)}_{\triangleq c}.
\end{equation}
It can be verified that $c < \log(e)$, and $c \rightarrow \log(e)$, as $\beta \rightarrow \infty$.
Thus, comparing (\ref{eq:achievable_fixed_tradeoff}) with (\ref{eq:above_bound}), the minimum feedback needed to
schedule $k$ active users in $b$ slots without collision is reduced from
essentially $\log(e)$ bits per user for the case of $b=k$ to $$\log(e)-c$$ bits
per user for $b>k$.  For example, if $\beta=2$, then $c=1$. We save 1 bit per 
active user.

In \figurename~\ref{fig:tradeoff}, the fixed-length achievabililty and converse
bounds are plotted for $n=10^6$ and $k=1000$. 
(The bounds for the variable-length code are omitted in the plot, because in this regime they are very close to the fixed-length bounds.)
At $b=1000$, the minimum
feedback rate needed to ensure collision-free scheduling is about 1457 bits, which is close to $\log(e)$ bits per user. 
If $b$ is increased to 2000, the amount of feedback required is reduced to 457
bits, a reduction of 1 bit per active user.
Increasing $b$ further reduces the feedback rate
even more, although eventually there is diminished return. 
In the context of the three-phase massive random access scheme, this means that the duration of the second phase where the BS uses the feedback message to schedule the users can be reduced at the expense of additional slots in the third phase.

\begin{figure}[t]
\centering
\includegraphics[width=\columnwidth]{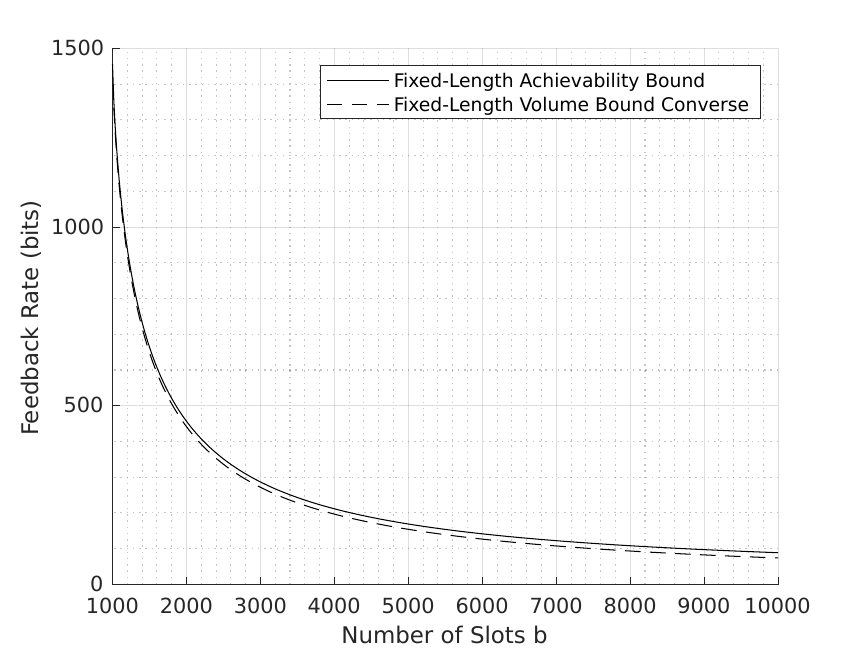}
\caption{Trade-off between the number of slots $b$ and the minimum collision-free
	feedback rate for $n=10^6$ and $k=1000$.}
\label{fig:tradeoff}
\end{figure}

 \section{Feedback Rate for $b<k$} \label{sec:multiuser}
In the previous sections we have analyzed the amount of feedback a BS must
provide in order to schedule $k$ active users into $b \geq k$ slots in 
a collision-free manner. In this section, we consider the case where only
$b < k$ slots are available. Clearly, multiple active users must now be
scheduled in the same slot, but if the BS receiver can tolerate up to 
$m$ users in each slot in the subsequent data transmission phase 
(e.g., by using a multiuser detector), we can define a similar
``collision-free'' condition. Let $m>1$ be an integer such that
\begin{equation}
(m-1)b < k \le mb.
\end{equation}
Our goal becomes to construct a feedback code such that no more than $m$ users
are scheduled in any single slot. Toward this end, we need to first modify 
the definition of a partition covering $\mathbf{C}\left( \part \right)$
in the collision-free condition \eqref{eq:no_err_condition} to the following:
\begin{equation}
\mathbf{C}\left( \part, m \right) = \\
 \left \{ \mathbf{A} \middle| \left \lvert \mathbf{A} \cap \mathbf{X}_i \right \rvert \leq m, \mathbf{A} \in \binom{[n]}{k}, i=1,\dotsc, b \right\}.
\end{equation}
Under this definition, an activity pattern is covered if no more than $m$ active user indices occupy each subset of the partition. Naturally the zero-error condition becomes
\begin{equation}\label{eq:no_err_m_users}
\bigcup \limits_{i=1}^{T} \mathbf{C} \left( \part^{(i)}, m \right) = \binom{[n]}{k}.
\end{equation}
Then, a collision-free encoder should satisfy
\begin{equation}
f(\mathbf{A})  =  t \;\;\text{s.t.}\;\; \mathbf{A} \in 
\mathbf{C}\left( \part^{(t)}, m \right) \label{eq:encoding_m_users}.
\end{equation}
Using these definitions and forming arguments similar to the case of $b=k$, we can establish upper and lower bounds on the minimum average feedback rate such that no more than $m$ of the $k$ active users access each of the $b$ transmission slots. The theorems below establish bounds on the minimum rate for both the fixed-length and variable-length codes. Proofs can be found in the appendix. 

\begin{thm}\label{thm:multi_slot_fixed_rate_achievable}
For  $b \leq k \leq n$, the minimum rate of a fixed-length collision-free feedback scheduling code $R^{*}_{f}(n,k,b)$
is upper bounded as
\begin{equation}
R^{*}_{f}(n,k,b) \leq b \left(\frac{1}{2}\log\left( 2 \pi m\right) + \frac{\log(e)}{12m} \right) \\ 
+ \log \left( \ln \left( \frac{n}{mb}\right) + 1 \right),
\end{equation}
where $m$ is an integer such that $(m-1)b < k \le mb$.
\end{thm}

\begin{thm}\label{thm:multi_slot_var_rate_achievable}
For  $b \leq k \leq n$, the minimum rate of a variable-length collision-free feedback scheduling code $R^{*}_{v}(n,k,b)$
is upper bounded as
\begin{equation}\label{eq:multi_slot_var_rate_achievable}
R^{*}_{v}(n,k,b) \leq  \left( b + \frac{1}{m}\right) \left( \frac{1}{2}\log(2\pi m) + \frac{\log(e)}{12m}\right),
\end{equation}
where $m$ is an integer such that $(m-1)b < k \le mb$.
\end{thm}

\begin{thm}\label{thm:multi_slot_var_rate_volume}
For  $b \leq k \leq n$ and $k=mb$, 
the minimum rate of a variable-length collision-free feedback scheduling code  $R^{*}_{v} \left(n,k,b \right)$
is lower bounded as
 \begin{multline}
R^*_{v}\left(n,k,b\right) 
\geq b \left( \frac{1}{2} \log\left( 2 \pi m \right)
+ \frac{\log(e)}{12m + 1}\right)
 - \frac{1}{2}\log \left( 2 \pi k\right) \\
  - \frac{\log(e)}{12 k}
 - \log(\gamma \left(n',k,b \right)),
\end{multline}
where  $n' =  \left \lfloor \frac{n}{k}\right \rfloor k$, and 
    $\gamma(n',k,b) =
    \prod_{l=1}^{m} \frac{\left(n'-lb\right)^b}{\left(
    n'-lb\right)^{\underline{b}}}$,
which goes to $1$ as $n$ goes to infinity.
The same bound also applies to the fixed-length code case, i.e., to $R^{*}_{f}(n,k,b)$.
\end{thm}

\begin{thm}\label{thm:multi_slot_fixed_rate_exclude}
For  $1 < b \leq k \leq n$, the minimum rate of a fixed-length collision-free feedback scheduling code $R^{*}_{f}(n,k,b)$
is lower bounded as
\begin{equation}
R^*_f(n,k,b) \geq \log \log \left( \frac{n}{k-1}\right) + \log(b) -1.
\end{equation}
\end{thm}

\begin{figure}[t]
\centering
\includegraphics[width=\columnwidth]{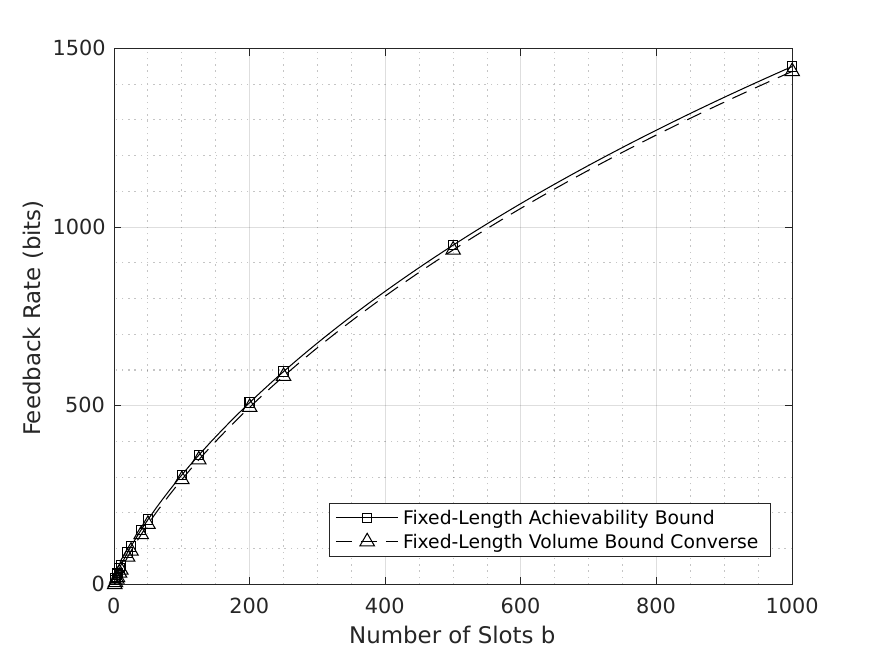}
	\caption{The minimum feedback rate needed to schedule $k=1000$ active users out of $n=10^6$ potential users into $b$ slots with no more than $m=\lceil \frac{k}{b}\rceil$ users in each slot using a fixed-length feedback code. Markers are placed at points where $\frac{k}{b}$ is an integer. Lines connecting the markers are interpolated bounds based on Theorems \ref{thm:multi_slot_fixed_rate_achievable} and \ref{thm:multi_slot_var_rate_volume}.}
\label{fig:multi_user_slots}
\end{figure}

Note that these bounds do not exactly coincide with bounds for $b=k$ when
$m=1$, as they appear to have different linear scaling coefficients in $k$.
But the difference is small. For achievability, the coefficient $\frac{1}{2} \log(2
\pi) + \frac{\log(e)}{12}$ differs from $\log(e)$ by about $0.003$. 
This minor difference can be
attributed to the differing approximations of the factorial terms. 

\figurename~\ref{fig:multi_user_slots} shows that the amount of feedback 
required to schedule $k$ users into $b=k/m$ slots with $m$ users per slot
decreases as $m$ increases and $b$ decreases. This is because when $b$
is smaller, less information is needed to specify which slot each user 
should be scheduled in. As the theorems above show, if $m$ is large, 
the saving is about a factor of $O(\log(m)/m)$ for each user.

This way of scheduling multiple users into the same scheduling slot is not
without practical concerns, however. 
When each user is given their own slot, the BS can easily determine which user
sent which message.  However, if $m$ users transmit their payload in the same
slot, even if the BS can decode the payloads of the $m$ users, without
additional identifying information, potentially embedded in the payload, the BS would be
unable match users to their payloads. For the set-partition based feedback
scheduling scheme, because the users do not know which other users are active
and which slots the other active users are being scheduled into, it may take 
up to $O\left(\log \left(n/b \right)\right)$ additional bits per user 
for the active users to identify themselves, in effect shifting the 
$O(\log(n))$ saving in the feedback stage to the payload.
For this reason, the analysis for the case of $m>1$ 
is most interesting when the BS has a means of distinguishing users
that transmit in the same slot.
 \section{Practical Codes via Perfect Hashing} \label{sec:practical}
As the minimum-rate collision-free feedback strategy is closely related to 
the perfect hashing problem, we can use perfect hash families to specify 
the encoding and decoding functions for feedback messages in massive 
random access and leverage the perfect hashing literature to design the
collision-free scheduling codes.
The primary focus of the perfect hashing literature has been to develop practical hash
families and functions for storage applications. These
practical hashing families are evaluated based on the criteria of
\textit{storage space, query complexity} and \textit{build complexity}. Storage
space is the average amount of information required to indicate which function from
the family is collision-free on a specified set of keys. Query complexity is the computational complexity required to
evaluate the hash function. Build complexity (or build time) is the computational complexity of finding a
function within the family which results in no collisions. These criteria
are also important in the context of massive random access; they
correspond to feedback rate, decoding complexity and encoding complexity
respectively of the feedback scheduling code. 

To the best of the authors' knowledge, there are no known practical perfect hashing schemes that 
approach the random coding achievability bound in average rate while maintaining
a sub-exponential build time in $k$. In terms of rate, in the perfect hashing literature, it is often assumed that
the total number of keys, $n$, is such that $\log(n) \ll k$.  This means that
both the achievability and converse bounds for both the fixed and the variable-length
cases are $O(k)$. Due to this, the metric used for storage space is ``bits
per key", which is the coefficient of linear scaling in $k$. For the case of $b=k$ the best known
construction with non-exponential build time is the boolean satisfiability (SAT) based
method presented in \cite{weaver2019constructing} which has an average storage
space of $1.83k$, as opposed to the $\log(e)k \approx 1.44k$ scaling shown to be achievable by
Theorem \ref{thm:fixed_achieve_bound}. This implementation has both a build
time of less than 1 second and a query time in the hundreds of nanoseconds for $k=2^{16}$ in
modern hardware, making it suitable for low-latency applications.

The achievability bounds of this paper show that source coding can significantly
reduce the average rate required for collision-free feedback in some regimes of $n$ and $k$.  There are several
practical hash family schemes that make use of this fact to reduce their average
rate. For example the \textit{Compressed Hash-Displace} (CHD) \cite{BotelhoCHD2009} algorithm takes what would
otherwise be an $O(k \log\log(k))$ algorithm in terms of storage complexity and during the \textit{compress} phase
reduces the memory requirement to $O(k)$. Although the coefficient of the linear scaling of the overall algorithm is $2.07$ for the case of $b=k$, which is higher than the optimal $\log(e)k \approx 1.44k$ scaling as in Theorem \ref{thm:achieve_source_coded}, the CHD algorithm has the advantage that it can 
be adapted to other values of $b$, the results of which are presented in Table \ref{tab:practical_linear_scaling}. 
These practical schemes can provide considerable feedback rate saving if
used for scheduling in the massive random access context. 

\begin{table}
\centering
\caption{Linear Scaling Factors of Hashing Algorithms}
\resizebox{!}{!}{\begin{tabular}{c c c}
\toprule
Method & Bits Per Key & Load Factor \\ [0.5ex] \midrule
Random Coding & 1.44 & \multirow{3}{5em}{$b=k$}\\
SAT & 1.83 &\\
CHD & 2.07 & \\ \midrule
Random Coding & 0.89 & \multirow{2}{5em}{$b=1.23k$}\\
CHD & 1.40 & \\ \midrule
Random Coding & 0.44 & \multirow{2}{5em}{$b=2k$}\\
CHD & 0.69 & \\
\bottomrule
& & \\
\end{tabular}
}
\label{tab:practical_linear_scaling}
\end{table}

 \section{Conclusions}
This paper considers the problem of finding minimum-rate feedback strategies
for scheduling $k$ out of $n$ users in $b$ slots for massive random access. Both
the cases of fixed and variable-length feedback codes are considered. For the
fixed-length case, the main contributions of this paper are the formulation of
this problem as a set-partition problem and in showing the connection of this
problem to the perfect hash function problem. By forming this connection, we
establish that the optimal feedback strategy for the $b=k$ case must have a rate
that scales linearly in $k$ with an additive double logarithmic term in $n$ using fixed-length codes. For the variable-length code case, we present a novel proof
which shows the existence of a collision-free feedback code with an average rate
which depends only linearly in $k$ and has no dependence on $n$. 
The linear scaling factor in both cases is $\log(e)$ bits (or 1 nat). 
This paper also extends the results to when $b \geq k$ slots are available 
and to the case when $b=\frac{k}{m}$ slots are available and no more 
than $m$ users can be scheduled per slot. In both cases, the minimum feedback 
rate is shown to scale linearly in $k$ but with a smaller scaling factor, plus
a double logarithmic term in $n$ for the fixed-length case. 
In conclusion, the minimum number of feedback bits needed for collision-free
scheduling for massive random access is essentially $\log(e)$ bits (or 1 nat)
per active user for the $b=k$ case and smaller for $b>k$ or $b<k$ cases. 
This theoretical limit is much lower than the naive feedback scheme that
requires $\log(n)$ bits per active user. Practical codes toward this limit 
have been devised in the perfect hashing literature, and can 
be used for scheduling in massive random access applications.  
\appendix 
\subsection{Proof of Proposition \ref{thm:trade_off_fixed_rate_achievable}}
Fix some $\mathbf{A} \in \binom{\left[ n\right]}{k}$. If
we choose a $b$-partition at random, the probability that $\mathbf{A}$ is
covered by $\partcover{}$ can be written as
\begin{equation}
\mathrm{Pr}\left( \mathbf{A} \in  \partcover{} \right)  =
\frac{b^{\underline{k}}}{b^k} \triangleq p.
\end{equation}
By Lemma \ref{lem:achieve_epsilon}, we have \eqref{eq:lt_eps}. Making the substitution of $p=\frac{b^{\underline{k}}}{b^k}$ and taking the logarithm of both sides yields
\begin{equation}
T^{*}(n,k,\epsilon) \leq \ln \left(\frac{1}{\epsilon}\binom{n}{k}\right)\left(\frac{b^{\underline{k}}}{b^k}\right).
\end{equation}
This means that any rate satisfying the below is achievable:
\begin{equation}
    R \geq \log \left( \ln\left( \binom{n}{k}\right) \frac{b^{k}}{b^{\underline{k}}} \right).
\end{equation}
Noting that $\binom{n}{k} < \frac{n^k}{k!}$, we have
\begin{equation}\label{eq:random_coding_with_tradeoff_unsimplified}
    R \geq k \log(b) - \log \left( b^{\underline{k}}\right) + \log \left(
    \ln\left( \frac{n^k}{k!}\right)\right).
\end{equation}
Using the fact that $\sqrt{2 \pi}x^{x+\frac{1}{2}}e^{-x} < x! < \sqrt{2 \pi}x^{x+\frac{1}{2}}e^{-x}e^{\frac{1}{12x}}$ for all positive integers, we show that for $b > k$:
\begin{IEEEeqnarray}{lRl}
b^{\underline{k}} & = & \frac{b!}{(b-k)!}\\
& > & \frac{\sqrt{2 \pi}b^{b + \frac{1}{2}}e^{-b}}{\sqrt{2 \pi}
      (b-k)^{b-k+\frac{1}{2}}e^{-(b-k)}e^{\frac{1}{12(b-k)}}} \\
& > & e^{-k} b^{b + \frac{1}{2}}(b-k)^{-b+k-\frac{1}{2}}e^{-\frac{1}{12(b-k)}}\\
& > & e^{-k} b^k \left(1-\frac{k}{b}\right)^{-b+k-\frac{1}{2}}e^{-\frac{1}{12(b-k)}}.
\end{IEEEeqnarray}
Now taking the logarithm and noting that $-\frac{1}{2} \log \left( 1- \frac{k}{b}\right) > 0$ and $-\frac{\log(e)}{12(b-k)} > -1$ since $b>k$, we find
\begin{equation}\label{eq:log_factorial_bound}
\log \left( b^{\underline{k}}\right) > -k \log(e) + k\log(b) \\
+ \left(-b+k\right) \log \left( 1 - \frac{k}{b} \right)- 1.
\end{equation}
Together with \eqref{eq:random_coding_with_tradeoff_unsimplified}, this implies the existence of a collision-free feedback code with rate
\begin{equation} \label{eq:min_achieveable_rate_tradeoff}
R > k \log(e) + \log \left(\ln \left(\frac{n}{k} \right)+1\right) \\
 + (b-k)\log \left(1-\frac{k}{b}\right) +1.
\end{equation}
Since any achievable rate must be greater than or equal to the minimal rate, the minimal rate must satisfy (\ref{eq:achievable_fixed_tradeoff}).
\hfill$\blacksquare$

\subsection{Proof of Theorem \ref{thm:trade_off_var_rate_achievable}}
The main difference in the case of $b > k$ slots is that we have
\begin{equation}
\mathrm{Pr} \left( \mathbf{A} \in \partcover{(t)}\right) = \frac{b^{\underline{k}}}{b^{k}}.
\end{equation}
Just as in the case of $b=k$, we use the entropy of a geometric distribution to bound the entropy of the minimum-entropy greedy encoder $H ( f_{\hashfamily^*} ( \mathbf{A} ) )$. The parameter of the geometric distribution is $ \frac{b^{\underline{k}}}{b^{k}}$, thus
\begin{multline}
H ( f_{\hashfamily^*}( \mathbf{A} ) ) \leq
- \frac{1}{1- \epsilon}\left(\log \left( \frac{b^{\underline{k}}}{b^k}\right) \right. \\ +
\left. \left( \frac{b^{k}}{b^{\underline{k}}}-1\right) 
\log \left( 1 - \frac{b^{\underline{k}}}{b^k}\right)\right).
\end{multline}
Using \eqref{eq:log_factorial_bound} and noting that 
$\left( 1 - \frac{1}{x}\right) \log \left( 1-x\right) < \log(e)$, $\forall x \in (0,1)$,
this expression becomes
\begin{multline}
H ( f_{\hashfamily^*} ( \mathbf{A} ) ) \leq 
\frac{1}{1- \epsilon} \bigg(  (k+1) \log(e)  \\
\left. + \left(b-k\right) \log \left( 1 - \frac{k}{b}\right)+ 1 \right).
\end{multline}
Following the same logic as the proof of Theorem \ref{thm:achieve_source_coded}, this shows that the rate on the right-hand side of the above expression is achievable. 
Finally, since $\epsilon$ can be arbitrary, taking $\epsilon \rightarrow 0$ completes the proof of \eqref{eq:trade_off_var_rate_achievable}.
\hfill$\blacksquare$

\subsection{Proof of Theorem \ref{thm:trade_off_var_rate_volume}}

As in the case of $b = k$, the number of activity patterns covered by a single partition is maximized when there are an equal number of indices in each subset, or as close to this as possible, given that the number of elements in each subset must be an integer value. In particular, we can show that
\begin{equation}
\left \lvert \partcover{(t)} \right \rvert \leq \binom{b}{k} \left(\frac{n}{b}\right)^{k}.
\end{equation}
Thus, since all partitions must be covered, $T^*(n,k,b)$ must satisfy
\begin{equation}
T^{*}(n,k,b) \geq \frac{\binom{n}{k}}{\binom{b}{k} \left( \frac{n}{b}\right)^k} =  \left(\frac{b^k}{b^{\underline{k}}} \right) \left( \frac{n^{\underline{k}}}{n^k}\right) = c_{\max}.
\end{equation}
Taking the logarithm and writing in terms of rate, we have
\begin{equation}\label{eq:b_gt_k_converse_unsimplified_f}
R^{*}_{f}(n,k,b) \geq \log \left( \frac{b^{k}}{b^{\underline{k}}}\right) - \log\left( \frac{n^{k}}{n^{\underline{k}}} \right).
\end{equation}
For the variable-length case, the same technique from the case of $b=k$ in Theorem \ref{thm:var_rate_vol} applies except with this new $c_{\max}$.
Taking the supremum of the bound \eqref{eq:log_p_vol_bound}, which is achieved by $Q(\mathbf{A}_i) = \binom{n}{k}^{-1}$ yields:
\begin{equation}\label{eq:b_gt_k_converse_unsimplified_v}
R^*_v(n,k,b) \geq
\log \left( \frac{b^{k}}{b^{\underline{k}}}\right) - \log \left( \frac{n^k}{n^{\underline{k}}}\right),
\end{equation}
exactly matching \eqref{eq:b_gt_k_converse_unsimplified_f}.
Finally, note that  $\sqrt{2 \pi}x^{x+\frac{1}{2}}e^{-x}e^{\frac{1}{12x+1}} < x!$ $ < \sqrt{2 \pi}x^{x+\frac{1}{2}}e^{-x}e^{\frac{1}{12x}}$ 
for all positive integers $x$. This shows that for all $b > k$,
\begin{IEEEeqnarray}{lRl}
b^{\underline{k}} & = & \frac{b!}{(b-k)!}\\
& < & \frac{\sqrt{2\pi}b^{b+ \frac{1}{2}}e^{-b}e^{\frac{1}{12b}}}{\sqrt{2 \pi} (b-k)^{\left(b-k + \frac{1}{2}\right)}e^{-(b-k)}e^{\frac{1}{12(b-k)}}}\\
& < & e^{-k}b^{k} \left( 1 - \frac{k}{b}\right)^{\left(b-k + \frac{1}{2}\right)}e^{\frac{1}{12b}-\frac{1}{12(b-k)}}\\
& < & e^{-k}b^{k} \left( 1 - \frac{k}{b}\right)^{\left(b-k + \frac{1}{2}\right)}.\label{eq:log_factorial_leq}
\end{IEEEeqnarray}
so we have
\begin{equation} \label{eq:falling_fact_bound}
\log \left( \frac{b^{k}}{b^{\underline{k}}}\right)  \geq
k \log(e) 
+ \left( b-k + \frac{1}{2}\right)\log \left(1 - \frac{k}{b} \right).
\end{equation}
Substituting \eqref{eq:falling_fact_bound} into \eqref{eq:b_gt_k_converse_unsimplified_f} and \eqref{eq:b_gt_k_converse_unsimplified_v} completes the proof.
\hfill$\blacksquare$
 \subsection{Proof of Theorem \ref{thm:multi_slot_fixed_rate_achievable}}
Fix some $\mathbf{A} \in \binom{[n]}{k}$. We show achievability for $k=mb$. The same coding scheme works equally for all $k \le mb$. Let a $b$-partition $\part$ be chosen at random. We wish to determine the probability
\begin{equation}
\mathrm{Pr} \left( \mathbf{A} \in \mathbf{C} \left( \part, m \right) \right) \triangleq p.
\end{equation}
There are a total of $b^k$ ways that the $k$ elements of $\mathbf{A}$ can be distributed across the $b$ subsets that define $\part$. The number of ways that the $k$ elements of $\mathbf{A}$ can be distributed among the $b$ subsets of $\part$ such that $\left \lvert \mathbf{A} \cap \mathbf{X}_i \right \rvert = m $ for $i=1, \dotsc, b$ is 
\begin{equation}
\binom{mb}{m} \binom{m(b-1)}{m} \cdots \binom{m}{m} = \frac{k!}{\left(m!\right)^b}.
\end{equation} Therefore, we have
\begin{equation}
p = \frac{k!}{b^k\left( m!\right)^b}.
\end{equation}
The rest of the proof follows the same logic as the proof of Proposition \ref{thm:fixed_achieve_bound}.
From \eqref{eq:lt_eps}, we see that a random family of $T$ partitions, where $T$ satisfies
\begin{equation}
T \geq \ln \left( \frac{1}{\epsilon}\binom{n}{k}
\right)\left(\frac{b^k \left(m!\right)^b }{k!}\right),
\end{equation}
has a probability of at least $1-\epsilon$ of satisfying \eqref{eq:no_err_m_users}. 
As the minimum number of partitions $T^*(n,k,b)$ satisfying \eqref{eq:no_err_m_users} must be less than or equal to any achievable $T$, we have 
\begin{equation}
T^*(n,k,b)  \leq   \ln \left(\binom{n}{k}
\right)\left(\frac{b^k \left(m!\right)^b }{k!}\right).
\end{equation}
Taking the logarithm, we find that the $R_{f}^{*}(n,k,b)$ must satisfy
\begin{equation}
R^{*}_{f}(n,k,b)  \leq  \log \left(\frac{b^k(m!)^b}{k!}\right)
+ \log \left( \ln  \binom{n}{k} \right) .
\end{equation} 
The first term in the above expression can be bounded by
\begin{IEEEeqnarray}{lRl}
\log \left(\frac{b^k(m!)^b}{k!}\right) & \leq & k \log \left( \frac{e}{m}\right) + \frac{k}{m} \log \left( m!\right) \\
& \leq & \frac{k}{m} \left(\frac{1}{2}\log\left( 2 \pi m\right) + \frac{\log(e)}{12m} \right),
\label{eq:log_p_inequality}
\end{IEEEeqnarray}
where we have used the fact that $k! > k^{k}e^{-k}$ for $k \geq 1$ in the first inequality and $m! < \sqrt{2 \pi} m^{m + \frac{1}{2}}e^{-m} e^{\frac{1}{12m}}$ for $m \geq 1$ in the second inequality. Finally, noting that $\binom{n}{k} \leq \frac{n^k}{k!}$, and $k =mb$ we have 
\begin{equation}
R^{*}_{f}(n,k,b) \leq b \left(\frac{1}{2}\log\left( 2 \pi m\right) + \frac{\log(e)}{12m} \right) \\
+ \log \left( \ln \left( \frac{n}{mb}\right) + 1 \right),
\end{equation}
which completes the proof.
\hfill$\blacksquare$

\subsection{Proof of Theorem \ref{thm:multi_slot_var_rate_achievable}}
We show achievability for $k=mb$. Since increasing $k$ can only 
increase the required rate, this also shows achievability for $k \le mb$. 
The main difference in this proof, as compared to the case of $b = k$ is that based on the argument in Proposition \ref{thm:multi_slot_fixed_rate_achievable}, we have
\begin{equation}
\mathrm{Pr}\left( \mathbf{A} \in \mathbf{C}\left( \part, m \right)\right) = \frac{k!}{(m!)^b b^k} \triangleq p.
\end{equation}
Following the same logic as in the proof of Theorem \ref{thm:achieve_source_coded}, we can use the entropy of a geometric distribution with parameter $p$ to bound the entropy of the minimum-entropy greedy encoder $H \left( f_{\hashfamily^*} \left( \mathbf{A}\right)\right)$. This results in the bound 
\begin{multline}
H\left( f_{\hashfamily^*}\left( \mathbf{A}\right)\right) \leq -\frac{1}{1- \epsilon}
\left( \log \left( \frac{k!}{(m!)^b b^k}\right) \right. \\
+\left. \left( \frac{(m!)^b b^k}{k!} - 1\right) \log \left( 1 - \frac{k!}{b^k (m!)^b}\right) \right).
\end{multline}
Using \eqref{eq:log_p_inequality} and noting that
\begin{multline}
 \left( \frac{(m!)^b b^k}{k!} - 1\right) \log \left( 1 - \frac{k!}{b^k (m!)^b}\right) \leq \\
 \left(\frac{1}{2}\log\left( 2 \pi m\right) + \frac{\log(e)}{12m} \right),
\end{multline}
we find that
\begin{equation}
H\left( f_{\hashfamily^*} \left( \mathbf{A}\right)\right) \leq \frac{1}{1- \epsilon} \cdot \left( b + \frac{1}{m}\right)\left( \frac{1}{2}\log(2\pi m) + \frac{\log(e)}{12m}\right),
\end{equation}
where we have substituted $k=mb$.
Finally, since $R^*_{v}(n,k,b)$, which is the minimal entropy over all collision-free encoders, must be less than $H\left( f_{\hashfamily^*} \left( \mathbf{A}\right)\right)$ and $\epsilon$ can be arbitrary, taking $\epsilon \rightarrow 0$ gives \eqref{eq:multi_slot_var_rate_achievable}.
\hfill$\blacksquare$

\subsection{Proof of Theorem \ref{thm:multi_slot_var_rate_volume}}
We begin by noting that $R^*_f(n,k,b)$ and $R^*_v(n,k,b)$ are non-decreasing functions of $n$. To this end, define $n' = \left \lfloor \frac{n}{k} \right \rfloor k$ so that we can bound $R^*(n,k,b)$ below by $R^*(n',k,b)$ for both the fixed and variable-length cases.

As in the case of $b=k$, we  note that the number of activity patterns covered by a single partition is maximized when there are an equal number of indices in each subset (or nearly so), so when $b = \frac{k}{m}$ we have 
\begin{equation}
\left \lvert \mathbf{C}\left( \part, m\right) \right \rvert \leq \binom{\frac{n'}{b}}{m}^{b} = c_{\max}.
\end{equation}
Thus, in order to cover all activity patterns, we must have
\begin{equation}
T^{*}(n,k,b) \geq \frac{\binom{n'}{k}}{\binom{\frac{n'}{b}}{m}^{b}}.
\end{equation}
Taking the logarithm and expressing in terms of rate, we find that
\begin{equation}\label{eq:m_slot_converse_f}
R^{*}_{f}(n',k,b) \geq \log \left( \frac{b^k \left( m!\right)^{b}}{k!}\right) - 
\log \left( \gamma \left(n',k,b \right)\right),
\end{equation}
where we have defined $\gamma(n',k,b) \triangleq \prod_{l=1}^{m} \frac{\left(n'-l{b}\right)^{b}}{\left( n'-l{b}\right)^{\underline{b}}}.$
For the variable-length case, the same technique from the case of $b=k$ in Theorem
\ref{thm:var_rate_vol} applies, except with this new $c_{\max}$. 
Taking the supremum of the bound \eqref{eq:log_p_vol_bound}, which is again achieved by $Q(\mathbf{A}_i) = \binom{n}{k}^{-1}$, yields
\begin{equation}\label{eq:m_slot_conv_var}
    R^*_{v}(n',k,b) \geq  \log\left( \frac{b^k(m!)^b}{k!}\right)
- \log\left(\gamma \left(n',k,b\right) \right),
\end{equation}
exactly matching the fixed-length case. 
Then, using the fact that $m! > \sqrt{2\pi} m^{m+ \frac{1}{2}}e^{-m}e^{\frac{1}{12m + 1}}$, and that $k! < \sqrt{2 \pi} k^{k + \frac{1}{2}} e^{-k}e^{\frac{1}{12k}}$, we have the following bound:
\begin{multline} \label{eq:log_p_inequality_gt}
 \log \left( \frac{b^k\left( m!\right)^b}{k!}\right) \geq \frac{k}{m} \left( \frac{1}{2} \log\left( 2 \pi m\right)  + \frac{\log(e)}{12m + 1}\right) \\
 - \frac{1}{2}\log \left( 2 \pi k\right) - \frac{\log(e)}{12 k}.
\end{multline}
Applying \eqref{eq:log_p_inequality_gt} to \eqref{eq:m_slot_converse_f} and \eqref{eq:m_slot_conv_var} 
and noting that $R^*(n,k,b) \geq R^{*}(n',k,b)$ for both the fixed and variable-length cases completes 
the proof.

Note that $\gamma(n',k,b)$ goes to $1$ as $n\rightarrow \infty$ for fixed $k$ and $b$, because 
\begin{equation}
\lim_{n \rightarrow \infty} \frac{\left( n'-lb\right)^b}{\left( n'-lb\right)^{\underline{b}}} = 1,
\end{equation} 
for each $l = 0, \dotsc, m-1$. 
For $k \ll n$, this says that the minimum feedback rate must be at least linear in $k$ for large $n$. 
\hfill$\blacksquare$

\subsection{Proof of Theorem \ref{thm:multi_slot_fixed_rate_exclude}}

The proof closely parallels the proof of Proposition \ref{thm:fixed_exclusion} for the case of $b=k$. There is at least one partition with $\left \lfloor \frac{n}{b}\right \rfloor$ elements. Following the same exclusion argument, we find
\begin{equation}
n \left( 1 - \frac{m}{k}\right)^{T} \leq k-1.
\end{equation}
This gives us the following bound on $T$:
\begin{equation}
T \geq \frac{\log(n) - \log(k-1)}{-\log \left( 1 - \frac{m}{k}\right)}.
\end{equation}
Writing this in terms of rate and noting that $-\log \left( 1 - \frac{m}{k}\right) \leq \frac{2m}{k}$ for $k > m$, we have
\begin{equation}
R^*_{f}(n,k,m) \geq \log \log \left( \frac{n}{k-1}\right) + \log(b) -1.
\end{equation}
\hfill$\blacksquare$

\bibliographystyle{IEEEtran}
\bibliography{IEEEabrv,references}

\end{document}